\documentclass[11pt]{article}

\usepackage[letterpaper,margin=1in]{geometry}
\usepackage[T1]{fontenc}
\usepackage{lmodern}
\usepackage{microtype}
\usepackage{amsmath,amssymb,amsthm,mathtools}
\usepackage{enumitem}
\usepackage{graphicx}
\usepackage{float}
\usepackage{booktabs,tabularx,makecell}
\usepackage[table]{xcolor}
\usepackage{tikz}
\usetikzlibrary{arrows.meta,calc,positioning,fit,backgrounds,shapes.geometric}
\usepackage[numbers,sort&compress]{natbib}
\usepackage{hyperref}
\hypersetup{
  hidelinks,
  pdftitle={Resolving the Nelson--Nguyen Conjecture},
  pdfsubject={Sparse oblivious subspace embeddings}
}
\usepackage{xurl}

\definecolor{accent}{RGB}{63,101,154}
\definecolor{layerfill}{RGB}{247,249,252}
\definecolor{zoomfill}{RGB}{238,244,252}

\tikzset{
  layer/.style={
    draw=black!65,
    fill=layerfill,
    rounded corners=1.2pt,
    minimum width=4.25cm,
    minimum height=0.78cm,
    inner sep=0pt
  },
  selected layer/.style={
    layer,
    draw=accent,
    fill=zoomfill,
    line width=0.95pt
  },
  matrix frame/.style={
    draw=accent,
    fill=zoomfill!45,
    rounded corners=1.2pt,
    line width=0.95pt
  },
  guide/.style={
    densely dashed,
    draw=black!55,
    line width=0.55pt
  },
  zoom line/.style={
    densely dotted,
    draw=accent,
    line width=0.95pt
  },
  dim/.style={
    <->,
    >=Stealth,
    draw=black!70,
    line width=0.55pt
  },
  nz/.style={
    fill=white,
    inner sep=1.3pt,
    font=\small
  },
  focus nz/.style={
    draw=accent,
    fill=white,
    rounded corners=1pt,
    line width=0.8pt,
    inner xsep=3.2pt,
    inner ysep=2pt,
    font=\small
  },
  panel title/.style={
    font=\small\bfseries,
    anchor=west
  }
}

\allowdisplaybreaks
\setlist[itemize]{leftmargin=2em,itemsep=0.3em,topsep=0.4em}
\setlist[enumerate]{leftmargin=2.4em,itemsep=0.3em,topsep=0.4em}

\newtheorem{theorem}{Theorem}[section]
\newtheorem{lemma}[theorem]{Lemma}
\newtheorem{proposition}[theorem]{Proposition}

\newtheorem{claim}[theorem]{Claim}
\newtheorem{definition}[theorem]{Definition}
\theoremstyle{remark}

\newcommand{\R}{\mathbb{R}}
\newcommand{\E}{\mathbb{E}}
\newcommand{\Pp}{\mathbb{P}}
\newcommand{\Id}{I_d}
\newcommand{\tr}{\operatorname{tr}}
\newcommand{\cc}{\operatorname{cc}}
\newcommand{\op}{\mathrm{op}}
\newcommand{\one}{\mathbf{1}}
\newcommand{\cW}{\mathcal{W}}
\newcommand{\cP}{\Pi}
\newcommand{\cL}{\Lambda}
\newcommand{\cI}{\mathcal{I}}
\newcommand{\cO}{\mathcal{O}}
\newcommand{\cT}{\mathcal{T}}
\newcommand{\cC}{\mathcal{C}}
\newcommand{\cG}{\mathcal{G}}
\newcommand{\cJ}{\mathcal{J}}

\newcommand{\tT}{\mathsf{T}}
\newcommand{\Sign}{\mathsf{Sign}}
\newcommand{\Hash}{\mathsf{Hash}}
\newcommand{\Gram}{\mathsf{Gram}}
\newcommand{\Cont}{\mathsf{Cont}}
\newcommand{\eps}{\varepsilon}

\title{The Nelson–Nguyen Conjecture via Mean-to-Moments Concentration}
\author{
  Tung Mai\\
  Adobe Research
  \and
  Anup Rao\\
  Adobe Research
}
\date{}

\begin{document}
\maketitle

\begin{abstract}
An oblivious subspace embedding (OSE) is a distribution over matrices that approximately preserves the squared Euclidean norm of every vector in any fixed low-dimensional subspace. We prove the Nelson--Nguyen conjecture: for every $0<\delta<1$, there exists a distribution that gives an OSE with
\[
    O\!\left(\frac{d+\log(1/\delta)}{\eps^2}\right)
    \quad\text{embedding dimension and}\quad
    s=O\!\left(\frac{\log(d/\delta)}{\eps}\right)
    \quad\text{column sparsity},
\]
with failure probability at most $\delta$.
We first bound the mean spectral error by a trace-moment argument, then upgrade this bound to the desired high-probability guarantee using concentration and resampling. ChatGPT--5.6--Pro was used in proving and writing the results of this manuscript.
\end{abstract}

\section{Introduction}

An $m\times n$ matrix $\Phi$ is an $\eps$-subspace embedding for a subspace $V\subseteq\R^n$ if
\[
    (1-\eps)\lVert x\rVert_2^2
    \le \lVert \Phi x\rVert_2^2
    \le (1+\eps)\lVert x\rVert_2^2
    \qquad (x\in V).
\]
A distribution over $m\times n$ matrices is an $(\eps,\delta,d)$-oblivious subspace embedding (OSE) if for every fixed $d$-dimensional subspace $V$, a draw from the distribution is an $\eps$-subspace embedding for $V$ with probability at least $1-\delta$. The column sparsity parameter is the maximum number of nonzeros in any column of a matrix in the support. The number of rows $m$ is also called the \emph{embedding dimension} of the OSE.

Thus $\Phi$ approximately preserves the Euclidean norm of all vectors in $V$ while mapping them into a space whose dimension $m$ may be much smaller than $n$. Such embeddings are a basic sketching primitive in randomized numerical linear algebra: they reduce large least-squares, low-rank approximation, and related problems to smaller ones while preserving the relevant geometry. When $\Phi$ has few nonzeros per column, products of the form $\Phi A$ can be computed quickly for matrices $A$ with $n$ rows. Therefore, both the embedding dimension and the column sparsity are important algorithmic parameters.

Nelson and Nguyen~\cite{NelsonNguyen2013} conjectured that, for every failure probability $\delta$, there exists a sparse oblivious subspace embedding distribution with embedding dimension $m$ and column sparsity $s$ satisfying
\[
    m=O\!\left(\frac{d+\log(1/\delta)}{\eps^2}\right),
    \qquad
    s=O\!\left(\frac{\log(d/\delta)}{\eps}\right).
\]
We prove the conjecture by showing that the SparseStack distribution achieves this tradeoff. Informally, SparseStack stacks $s$ independent one-sparse CountSketch layers and normalizes the resulting matrix by $s^{-1/2}$.

\begin{theorem}[Nelson--Nguyen Conjecture]\label{thm:main}
For every $n\ge d\ge2$, $0<\eps\le1$, and $0<\delta<1$, there exists a
distribution over $m\times n$ matrices with
\[
    m=O\!\left(\frac{d+\log(1/\delta)}{\eps^2}\right)
    \quad\text{rows and}\quad
    s=O\!\left(\frac{\log(d/\delta)}{\eps}\right)
    \quad\text{nonzeros per column},
\]
such that, for every fixed $d$-dimensional subspace
$V\subseteq\R^n$, a matrix $\Phi$ sampled from the distribution satisfies
\[
    \Pp \!\left\{
        (1-\eps)\lVert x\rVert_2^2
        \le \lVert \Phi x\rVert_2^2
        \le (1+\eps)\lVert x\rVert_2^2
        \quad\text{for every }x\in V
    \right\}
    \ge 1-\delta.
\]
\end{theorem}

Nelson and Nguyen~\cite{NelsonNguyen2013} introduced Oblivious Sparse Norm-Approximating Projections (OSNAPs), a class of sparse embeddings in which every column has exactly $s$ signed entries of magnitude $s^{-1/2}$ and the support satisfies a suitable negative-correlation condition. They gave two concrete support distributions. In the first, the $s$ nonzero locations in each column form a uniformly random $s$-subset of the $m$ rows. In the second, the rows are partitioned into $s$ groups and one location is selected in each group. SparseStack is the fully independent instance of the second distribution, with the groups identified with its $s$ CountSketch layers. The SparseStack construction was analyzed earlier by Kane and Nelson in their work on sparse Johnson--Lindenstrauss transforms~\cite{KaneNelson2012}. The high-probability version of the Nelson--Nguyen conjecture for SparseStack was also listed as an open problem in~\cite{AmselEtAl2026}.

\paragraph{Technical overview.}
We prove Theorem~\ref{thm:main} using the SparseStack construction.
The proof has two stages. First, let \(U\) have orthonormal columns
spanning \(V\), and set
\[
E = U^\top \Phi^\top \Phi U - I_d.
\]
Then \(\lVert E\rVert_{\mathrm{op}}\) is exactly the maximum absolute
deviation of \(\lVert \Phi x\rVert_2^2\) from \(1\) over unit vectors
\(x \in V\). We control
\(\mathbb{E}\lVert E\rVert_{\mathrm{op}}\) by a trace-moment argument: for an
even \(\ell = \Theta(\log d)\), we bound
\(\mathbb{E}\operatorname{tr}(E^\ell)\).
The main challenge is that, in the resulting expansion, the randomness of the sketch is intertwined with the geometry of the subspace $V$.
We reorganize these terms
and introduce combinatorial objects so that these two effects can be bounded together.

We then upgrade the mean bound to a high-probability guarantee. For a fixed subspace \(V\), we call the largest decrease of \(\|\Phi x\|_2^2\) below \(1\), over unit vectors \(x\in V\), the lower distortion, and the largest increase above \(1\) the upper distortion. In our proof, lower distortion can be controlled one layer at a time, while upper distortion requires resampling individual random choices in SparseStack. Combining the two bounds shows that $\Phi$ preserves the squared norm of every vector in $V$ with the desired probability.

\paragraph{Further related work.}
Clarkson and Woodruff introduced sparse oblivious subspace embeddings in input-sparsity-time algorithms~\cite{ClarksonWoodruff2013}. For the CountSketch distribution, Nelson and Nguyen's optimized analysis gives $m=O(d^2/(\eps^2\delta))$ with failure probability at most $\delta$~\cite{NelsonNguyen2013}. For the OSNAP family, Nelson and Nguyen obtained $m=O(d\log^8 d/\eps^2)$ and $s=O(\log^3 d/\eps)$. They also gave the incomparable tradeoff $m=O(d^{1+\gamma}/\eps^2)$ and $s=O(1/\eps)$ for every fixed $\gamma>0$~\cite{NelsonNguyen2013}. Bourgain, Dirksen, and Nelson subsequently obtained $m=O(d\log^2 d/\eps^2)$ with $s=O(\log^4 d/\eps^2)$~\cite{BourgainDirksenNelson2015}, while Cohen reached the conjectured column sparsity with one remaining logarithmic factor in the embedding dimension~\cite{Cohen2016}.

Recent work first attained the optimal embedding dimension with larger column sparsity~\cite{HogsgaardEtAl2024,ChenakkodEtAl2024}, then reduced the column sparsity to near-optimal~\cite{ChenakkodDerezinskiDong2025}, and finally reached both target parameters up to subpolylogarithmic factors~\cite{ChenakkodDerezinskiDong2026}. Table~\ref{tab:sparse-ose-progress} summarizes these bounds.

\begin{table}[H]
\centering
\caption{Selected sparse oblivious subspace-embedding upper bounds.}
\label{tab:sparse-ose-progress}
\scriptsize
\setlength{\tabcolsep}{2.8pt}
\renewcommand{\arraystretch}{1.2}
\begin{tabularx}{\textwidth}{@{}>{\raggedright\arraybackslash}p{0.18\textwidth}>{\centering\arraybackslash}p{0.200\textwidth}>{\centering\arraybackslash}p{0.200\textwidth}>{\centering\arraybackslash}p{0.105\textwidth}>{\raggedright\arraybackslash}X@{}}
\toprule
Result & Embedding dimension $m$ & Column sparsity $s$ & Failure probability & Main distinction \\
\midrule
\rowcolor{gray!10}
{\bfseries Nelson--Nguyen conjecture~\cite{NelsonNguyen2013}}
    & {\boldmath\(O((d+\log(1/\delta))/\eps^2)\)}
    & {\boldmath\(O(\log(d/\delta)/\eps)\)}
    & {\boldmath\(\delta\)}
    & \textbf{Conjectured target tradeoff.} \\
\addlinespace[1pt]
CountSketch~\cite{NelsonNguyen2013,ClarksonWoodruff2013,MengMahoney2013}
    & $O(d^2/(\eps^2\delta))$
    & $1$
    & $\delta$
    & One nonzero per column; linear $1/\delta$ dependence in $m$. \\
Nelson--Nguyen~\cite{NelsonNguyen2013}
    & $O(d\log^8 d/\eps^2)$
    & $O(\log^3 d/\eps)$
    & $\le1/3$
    & Original OSNAP bound near both targets. \\
Bourgain--Dirksen--Nelson~\cite{BourgainDirksenNelson2015}
    & $O(d\log^2 d/\eps^2)$
    & $O(\log^4 d/\eps^2)$
    & constant
    & Improved embedding dimension in the sparse-JL line. \\
Cohen~\cite{Cohen2016}
    & $O(d\log(d/\delta)/\eps^2)$
    & $O(\log(d/\delta)/\eps)$
    & $\delta$
    & Conjectured column sparsity with one extra logarithm in $m$. \\
H{\o}gsgaard et al.\ (2024) \cite{HogsgaardEtAl2024}
    & $O(d/\varepsilon^2)$
    & \makecell{
        $O((d/\log(1/\varepsilon)$\\[-1pt]
        $\quad{}+d^{2/3}\log^{1/3}d)/\varepsilon)$
    }
    &
    $2^{-d^{2/3}}$
    &
First optimal $m$ with column sparsity $o(d/\varepsilon)$.
\\
Chenakkod et al.\ (2024)~\cite{ChenakkodEtAl2024}
    & $O((d+\log(1/\delta))/\eps^2)$
    & $O(\log^4(d/\delta)/\eps^6)$
    & $\delta$
    & First optimal $m$ with polylogarithmic column sparsity. \\
Chenakkod--Derezi\'nski--Dong (2025)~\cite{ChenakkodDerezinskiDong2025}
    & $O((d+\log(1/\delta))/\eps^2)$
    & \makecell[c]{$O(\log^2(d/(\eps\delta))/\eps)$\\[-1pt]$+\,O(\log^3(d/(\eps\delta)))$}
    & $\delta$
    & Optimal $m$ with near-optimal column sparsity. \\
Chenakkod--Derezi\'nski--Dong (2026)~\cite{ChenakkodDerezinskiDong2026}
    & \makecell[c]{$O((d+\log(d/\delta))$\\[-1pt]$\log^{o(1)}(d/\delta)/\eps^2)$}
    & $O(\log^{1+o(1)}(d/\delta)/\eps)$
    & $\delta$
    & Both targets up to subpolylogarithmic factors. \\
Heidary (2026)~\cite{Heidary2026} (concurrent)
    & $O((d+\log(1/\delta))/\eps^2)$
    & $O(\log(d/\delta)/\eps)$
    & $\delta$
    & Finite tensor-product bounds for high trace moments. \\
\rowcolor{gray!16}
\textbf{This paper}
    & {\boldmath\(O((d+\log(1/\delta))/\eps^2)\)}
    & {\boldmath\(O(\log(d/\delta)/\eps)\)}
    & {\boldmath\(\delta\)}
    & \textbf{Combinatorial mean bound and mean-to-moments concentration.} \\
\bottomrule
\end{tabularx}
\end{table}

Some earlier results are applicable directly to SparseStack. They either analyzed the exact SparseStack distribution or considered broader sparse embedding families that contained it. H{\o}gsgaard et al.~\cite{HogsgaardEtAl2024} and Chenakkod, Derezi\'nski, and Dong~\cite{ChenakkodDerezinskiDong2026} analyzed SparseStack directly. Chenakkod et al.~\cite{ChenakkodEtAl2024} studied the same distribution as an OSNAP with jointly independent subcolumns. The broader families generalized different aspects of the construction. Nelson and Nguyen~\cite{NelsonNguyen2013} and Cohen~\cite{Cohen2016} analyzed OSNAP families with fixed column sparsity and a global negative-correlation condition on the support indicators. Chenakkod, Derezi\'nski, and Dong~\cite{ChenakkodDerezinskiDong2025} retained the one-nonzero-per-subcolumn structure but required only $O(\log(d/(\varepsilon\delta)))$-wise independence in the signs and locations. Bourgain, Dirksen, and Nelson~\cite{BourgainDirksenNelson2015} studied a different class of sparse Johnson--Lindenstrauss transforms with exactly $s$ nonzeros per column and negatively correlated support indicators within each column. They assumed that supports were independent across columns and that signs were independent Rademacher variables.

Tropp's comparison theorem~\cite{Tropp2026} gave a lower distortion guarantee result for complex Steinhaus SparseStack. In contrast, Theorem~\ref{thm:main} establishes the full two-sided distortion guarantee for the real Rademacher SparseStack. Tikhomirov developed a level-set entropy method for the upper spectral edge of sparse random matrices
\cite{Tikhomirov2026LevelSetEntropy}. The method applies to a broader
class of models with negatively associated support masks, including
Rademacher SparseStack. For SparseStack, it gives a constant-factor upper-distortion bound
with high probability using $m = O\left(d(\log\log d)^2\right)$, and $s=O(\log m)$.

On the lower-bound side, Nelson and Nguyen established the optimal general dependence of the embedding dimension and initiated lower bounds relating column sparsity and embedding dimension~\cite{NelsonNguyen2014}. Li and Liu proved sharp limitations for one-sparse embeddings and subsequently strengthened these tradeoffs~\cite{LiLiu2022,LiLiu2026}.

The trace-moment expansion approach follows sparse Johnson--Lindenstrauss and OSNAP analyses~\cite{KaneNelson2012,BravermanOstrovskyRabani2010,NelsonNguyen2013}. The walk encodings are related to sparse random-matrix moment methods~\cite{BauerGolinelli2001,HainzlPanafieu2026,BenaychGeorgesEtAl2020}. The Eulerian contraction used in Section~\ref{sec:gram} isolates a mechanism already present in the OSNAP moment argument. The form used in this paper allows vertex-dependent label restrictions and achieves the factor $d$ per connected component explicitly.

\paragraph{Concurrent work.}
Independently, Heidary~\cite{Heidary2026} established the same
parameter tradeoff for the fully independent SparseStack distribution. The proof couples each signed one-hot selector to a vector with independent
three-point entries and directly controls arbitrary even trace moments through
a finite tensor-product operator representation. Our proof instead bounds the
mean spectral error using injective Eulerian Gram sums and a cycle-rank-weighted
count of canonical walks, and then upgrades the mean bound using layerwise
Bousquet concentration and coordinate resampling.

\section{Preliminaries}\label{sec:preliminaries}

\subsection{Notation}

For a positive integer \(k\), we write
\[
[k] := \{1,\ldots,k\}.
\]
For a finite set \(A\), its cardinality is denoted by \(|A|\).
\(\mathbf{1}\{\mathcal{E}\}\) denotes the indicator of an event or
statement \(\mathcal{E}\).

For vectors \(x,y\), we write \(\langle x,y\rangle\) for the Euclidean
inner product and \(\|x\|_2\) for the Euclidean norm. For a matrix \(M\),
we write \(M^\top\) for its transpose, \(\operatorname{tr}(M)\) for its
trace, and \(\|M\|_{\mathrm{op}}\) for its operator norm. The
\(r\times r\) identity matrix is denoted by \(I_r\). If \(M\) is
symmetric, then \(\lambda_{\min}(M)\) and \(\lambda_{\max}(M)\) denote
its smallest and largest eigenvalues, respectively, and \(M\succeq 0\)
means that \(M\) is positive semidefinite.

For a real-valued random variable \(X\) and \(q\geq 1\), let
\[
\|X\|_q := \bigl(\mathbb{E}|X|^q\bigr)^{1/q}.
\]
For \(a\in\mathbb{R}\), define its positive and negative parts by
\[
(a)_+ := \max\{a,0\},
\qquad
(a)_- := \max\{-a,0\}.
\]

A partition \(\pi\) of a set \(\Omega\) is a collection of pairwise
disjoint nonempty subsets whose union is \(\Omega\). The elements of
\(\pi\) are called its blocks. For \(z\in\Omega\), we write
$
[z]_\pi
$
for the unique block of \(\pi\) containing \(z\). In other words, \([z]_\pi\) is the equivalence class of \(z\) under the equivalence
relation induced by \(\pi\). We write \(|\pi|\) for the number of
blocks of \(\pi\).

For a graph or multigraph \(H\), we write \(V(H)\) and \(E(H)\) for its
vertex and edge sets, \(\deg_H(v)\) for the degree of a vertex \(v\),
and \(\operatorname{cc}(H)\) for the number of connected components of
\(H\). A loop contributes two to the degree of its incident
vertex.

\subsection{SparseStack construction}

For positive integers $s$ and $B$, define $\operatorname{SparseStack}(s,B)$ as follows.
For each layer $\gamma\in[s]$, independently sample $h_\gamma(i)$ uniformly from $[B]$ for every $i\in[n]$, and sample independent Rademacher signs $\sigma_{\gamma i}\in\{-1,1\}$. Define $S\in\R^{sB\times n}$ by
\begin{equation*}
    S_{(\gamma,b),i}
    =\sigma_{\gamma i}\one\{h_\gamma(i)=b\},
    \qquad
    \Phi=s^{-1/2}S.
\end{equation*}
Figure~\ref{fig:sparsestack-schematic} illustrates the construction.

\begin{figure}[H]
  \centering
  \resizebox{\linewidth}{!}{\begin{tikzpicture}[font=\small]

    \node[panel title] at (1.30,5.72) {(a) SparseStack};

    \node[font=\large] at (0.82,2.52) {$S=$};
    \draw[dim] (1.30,5.13) -- node[fill=white,inner sep=1.5pt,above=2pt] {$n$ columns} (5.55,5.13);
    \draw[dim] (0.17,0.22) -- (0.17,4.86);
    \node[anchor=east,align=right,fill=white,inner sep=1.5pt] at (0.00,2.54)
      {$s$ layers\\[-1pt]$sB$ rows};

    \node[layer]          (L1) at (3.425,4.48) {$S_{1}\in\mathbb{R}^{B\times n}$};
    \node                 at (3.425,3.58) {$\vdots$};
    \node[selected layer] (Lg) at (3.425,2.52) {$S_{\gamma}\in\mathbb{R}^{B\times n}$};
    \node                 at (3.425,1.48) {$\vdots$};
    \node[layer]          (Ls) at (3.425,0.61) {$S_{s}\in\mathbb{R}^{B\times n}$};

    \node[panel title] at (8.05,5.72) {(b) One layer};

    \coordinate (Msw) at (8.05,0.70);
    \coordinate (Mnw) at (8.05,4.67);
    \coordinate (Mne) at (14.90,4.67);

    \draw[zoom line] (Lg.north east) -- (Mnw);
    \draw[zoom line] (Lg.south east) -- (Msw);

    \draw[matrix frame] (Msw) rectangle (Mne);

    \node at (8.62,4.92) {$1$};
    \node at (9.62,4.92) {$2$};
    \node at (10.62,4.92) {$3$};
    \node at (12.10,4.92) {$i$};
    \node at (13.35,4.92) {$\cdots$};
    \node at (14.45,4.92) {$n$};

    \node[anchor=east,fill=white,inner sep=1pt] at (7.82,4.18) {$1$};
    \node[anchor=east,fill=white,inner sep=1pt] at (7.82,3.35) {$2$};
    \node[anchor=east,fill=white,inner sep=1pt] at (7.82,2.80) {$\vdots$};
    \node[anchor=east,fill=white,inner sep=1pt] at (7.82,2.08) {$h_{\gamma}(i)$};
    \node[anchor=east,fill=white,inner sep=1pt] at (7.82,1.52) {$\vdots$};
    \node[anchor=east,fill=white,inner sep=1pt] at (7.82,1.08) {$B$};

    \draw[guide] (12.10,4.67) -- (12.10,2.36);
    \draw[guide] (12.10,1.80) -- (12.10,0.70);
    \draw[guide] (8.05,2.08) -- (11.75,2.08);
    \draw[guide] (12.45,2.08) -- (14.90,2.08);

    \node[nz]       at (8.62,4.18) {$+1$};
    \node[nz]       at (10.62,3.35) {$+1$};
    \node[nz]       at (14.45,3.35) {$-1$};
    \node[nz]       at (9.62,0.98) {$-1$};
    \node[focus nz] at (12.10,2.08) {$\sigma_{\gamma i}$};

    \draw[dim] (8.05,0.32) -- node[fill=white,inner sep=1.5pt,below=2pt] {$n$ columns} (14.90,0.32);
    \draw[dim] (15.25,0.70) -- node[fill=white,inner sep=1.5pt,right=3pt] {$B$ rows} (15.25,4.67);
    \node at (11.475,-0.4)
      {$S_{(\gamma,b),i} =\sigma_{\gamma i}\one\{h_\gamma(i)=b\}$};

    \end{tikzpicture}
  }
  \caption{SparseStack $S$ is formed by vertically stacking $s$ independent $B\times n$ CountSketch layers. The highlighted layer is enlarged on the right. In layer $\gamma$, column $i$ has one nonzero, located in row $h_{\gamma}(i)$ and equal to $\sigma_{\gamma i}$.}
  \label{fig:sparsestack-schematic}
\end{figure}
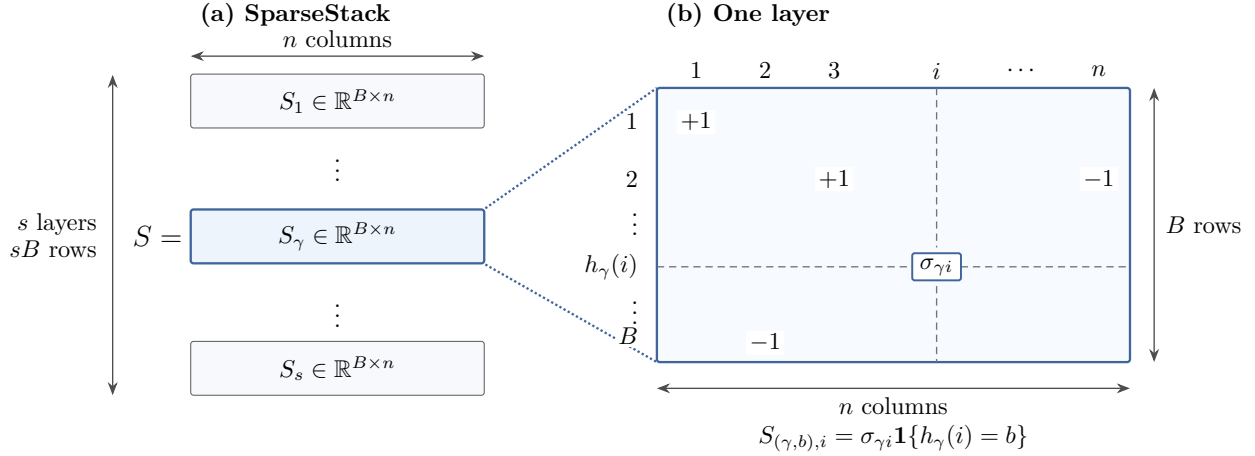

Every column has one nonzero in each layer and therefore exactly $s$ nonzeros. For a layer--bucket pair $\beta=(\gamma,b)$, let
\begin{equation*}
    I_\beta=\{i\in[n]:h_\gamma(i)=b\}.
\end{equation*}

\section{Proof overview}
\label{sec:overview}

Fix a $d$-dimensional subspace $V\subseteq\R^n$, and let $U\in\R^{n\times d}$ have orthonormal columns spanning $V$. Let
\[
    E=U^\top\Phi^\top\Phi U-\Id.
\]
The subspace-embedding guarantee is equivalent to $\lVert E\rVert_{\op}\le\eps$.

The proof separates the mean bound from a mean-to-moments
concentration argument. First, a trace-moment argument shows that, with
\[
    R:=\lVert E\rVert_{\op},
    \qquad
    \mu:=\E R,
\]
we have $\mu=O(\eps)$. Independently, we prove the modular bound
\[
    \lVert R\rVert_q
    =O\left(
        \mu+\sqrt{\frac qm}+\frac qs
    \right),
    \qquad q\ge4\log(2d^2).
\]
Combining the two bounds and taking
$q=\Theta(\log(d/\delta))$ gives the desired failure probability.

For the mean bound, let $0<\eps\le1$, let $p>0$ be an auxiliary
parameter, and assume that
\begin{equation}\label{eq:technical-regime}
    p\ge\log d,
    \qquad
    \frac{p}{8\eps}\le s\le\frac{p}{4\eps},
    \qquad
    B\ge\frac{8d}{\eps p},
    \qquad
    m=sB.
\end{equation}
The final section verifies these assumptions for concrete parameter choices used in Theorem~\ref{thm:main}.
The mean-to-moments proposition below holds for arbitrary positive
integers $s$ and $B$ and does not require \eqref{eq:technical-regime}.

\subsection{Mean spectral error}
Our first goal is to control the mean spectral error.
\begin{proposition}[Mean spectral error]\label{prop:mean-error}
There exists a universal constant $C_1$ such
that
\[
    \E\lVert E\rVert_{\op}\le C_1\eps.
\]
\end{proposition}

\paragraph{Trace expansion.}
Using the notation introduced in Section~\ref{sec:preliminaries},
where $S$ is the unnormalized SparseStack matrix, $s$ is the number of layers,
$h_\gamma$ are the hash functions, and $\sigma_{\gamma i}$ are the signs, let
\[
    G=sE=U^\top(S^\top S-sI_n)U.
\]
To control $\E\lVert E\rVert_{\op}$, we choose an even integer $\ell=\Theta(\log d)$ and show that
\[
    \E\tr(G^\ell)\le d\,(Cp)^\ell
\]
where $C$ is a universal constant. Since $G$ is real symmetric, this implies
\[
    \E\lVert E\rVert_{\op}
    \le \frac1s\bigl(\E\tr G^\ell\bigr)^{1/\ell}
    =O\!\left(\frac{p}{s}\right) = O(\varepsilon).
\]

Write $u_i^\top$ for the $i$th row of $U$. Expanding $\tr(G^\ell)$, each term is indexed by a \emph{trace tuple} $\tau=((\gamma_t,i_t,j_t))_{t=1}^{\ell}$, with $i_t\ne j_t$:
\[
X_\tau
=
\underbrace{
    \prod_{t=1}^{\ell}
        \sigma_{\gamma_t i_t}\sigma_{\gamma_t j_t}
}_{\text{sign factor}}
\cdot
\underbrace{
    \prod_{t=1}^{\ell}
        \one\!\left\{
            h_{\gamma_t}(i_t)=h_{\gamma_t}(j_t)
        \right\}
}_{\text{hash factor}}
\cdot
\underbrace{
    \prod_{t=1}^{\ell}
        \left\langle u_{j_t},u_{i_{t+1}}\right\rangle
}_{\text{Gram factor}},
\qquad i_{\ell+1}=i_1.
\]
Each $X_\tau$ contains three factors corresponding to the random signs, the required hash collisions, and inner products of the rows of $U$.
The sign factor has expectation zero unless every layer--coordinate pair $(\gamma,i)$ occurs an even number of times among the $2\ell$ trace positions. When this condition holds, every sign has an even power, so the sign factor is equal to $1$. Thus, for the surviving terms, it remains to analyze the hash and Gram factors.

To organize these terms, let $\pi(\tau)$ be a \emph{partition} obtained by partitioning the $2\ell$ \emph{trace positions} of $\tau$
\[
    (\gamma_1,i_1),(\gamma_1,j_1),\ldots,
    (\gamma_\ell,i_\ell),(\gamma_\ell,j_\ell)
\]
according to the equality relation. We call the equivalence classes the blocks of $\pi(\tau)$.
We may
reorganize the trace expansion as
\[
    \E\tr(G^\ell)
    =
    \sum_{\pi}
    \underbrace{
        \sum_{\tau:\,\pi(\tau)=\pi}
        \E X_\tau
    }_{\displaystyle \Cont(\pi)}.
\]
This reorganization groups many differently labeled tuples into a single partition and exposes the underlying combinatorial structure. After discarding the tuples eliminated by the random signs, only partitions with every block of even size will survive. We call them \emph{admissible partitions}.

\paragraph{Contribution of a fixed partition} For a fixed admissible partition \(\pi\), we separately consider the hash and Gram factors by introducing two multigraphs, both with the blocks of \(\pi\) as vertices. For each \(t\), the \textit{collision multigraph} \(\mathcal{C}_\pi\) has an edge between the blocks containing the two trace positions \((\gamma_t,i_t)\) and \((\gamma_t,j_t)\). This edge represents the condition
$ h_{\gamma_t}(i_t)=h_{\gamma_t}(j_t). $

The \textit{Gram multigraph} \(\mathcal{G}_\pi\) has an edge between the blocks containing \((\gamma_t,j_t)\) and \((\gamma_{t+1},i_{t+1})\), with $i_{\ell+1}=i_1$. This edge represents the factor
 $\langle u_{j_t},u_{i_{t+1}}\rangle. $

Let \(\operatorname{cc}(H)\) denote the number of connected components of a graph \(H\). The degree of a block in either \(\mathcal{C}_\pi\) or \(\mathcal{G}_\pi\) equals the number of trace positions in that block. \space
Only partitions whose blocks all have even size survive, because otherwise the sign factor has expectation zero. For such partitions, every vertex of both multigraphs has even degree.

Once the partition \(\pi\) is fixed, a trace tuple \(\tau\) inducing \(\pi\) is obtained by labeling each block \(D\in\pi\) with a layer--coordinate pair
\[
 (\lambda(D),a(D))\in[s]\times[n],
 \]
 with distinct blocks receiving distinct pairs. Every trace position in \(D\) then carries the label \((\lambda(D),a(D))\). Each edge of \(\mathcal C_\pi\) joins two blocks whose trace positions share the same layer index \(\gamma_t\). Hence \(\lambda\) is equal at the endpoints of every edge and therefore constant on each connected component of \(\mathcal C_\pi\). Each component may be assigned any layer in \([s]\), independently of the others, giving
\(
 s^{\operatorname{cc}(\mathcal C_\pi)}
 \)
 possible layer assignments. Moreover, since distinct blocks must represent distinct pairs \((\lambda(D),a(D))\), whenever two blocks are assigned to the same layer, they must receive different coordinates. We call this the \emph{per-layer injectivity constraint}.

Now consider a connected component of \(\mathcal{C}_\pi\) containing \(b\) blocks. These blocks have distinct coordinate labels in a common layer. Each edge contributes an indicator that its endpoint coordinate labels hash to the same bucket. By connectivity, their product equals one exactly when all \(b\) coordinate labels hash to one bucket. This event has probability \(B^{1-b}\). Multiplying over the components gives
\[
B^{\operatorname{cc}(\cC_\pi)-|\pi|}.
\]

The remaining task is to sum the Gram factors represented by
\(\mathcal{G}_\pi\). We call such an expression an \emph{Eulerian Gram sum}: it is a sum
over labelings of the vertices of an Eulerian multigraph, with each
labeling weighted by the product of the Gram factors
\(\langle u_{a(v)},u_{a(w)}\rangle\) over its edges. First consider the unrestricted sum over all coordinate
assignments \(a:\pi\to[n]\). Lemma~\ref{lem:restricted-contraction} gives
\[
\left|
\sum_{a:\pi\to[n]}
\prod_{\{D,D'\}\in E(\mathcal{G}_\pi)}
\langle u_{a(D)},u_{a(D')}\rangle
\right|
\le d^{\operatorname{cc}(\mathcal{G}_\pi)}.
\]

The tuples inducing \(\pi\), however, correspond only to assignments that are
injective within each layer. The unrestricted bound does not immediately apply
to this restricted sum, because the Gram products can have either sign:
removing non-injective assignments may remove cancellations. Our first main
analytic lemma shows that enforcing the constraint costs only an exponential
factor in the number of blocks:
\[
\left|
\sum_{\substack{
a:\pi\to[n]\\
a\text{ is injective within each layer}
}}
\prod_{\{D,D'\}\in E(\mathcal{G}_\pi)}
\langle u_{a(D)},u_{a(D')}\rangle
\right|
\le
(2e)^{|\pi|}d^{\operatorname{cc}(\mathcal{G}_\pi)}.
\]

In particular, the per-layer injectivity constraint costs a factor of at most
\((2e)^{|\pi|}\), while preserving the exponent
\(\operatorname{cc}(\mathcal G_\pi)\) of \(d\). Combining the number of layer
assignments, the hash-collision probability, and the restricted Gram-sum bound,
we obtain
\[
|\Cont(\pi)|
\le
s^{\operatorname{cc}(\mathcal C_\pi)}
B^{\operatorname{cc}(\mathcal C_\pi)-|\pi|}
(2e)^{|\pi|}
d^{\operatorname{cc}(\mathcal G_\pi)}.
\]
Using \eqref{eq:technical-regime}, we can show that
\[
|\Cont(\pi)|
\le
d^{\operatorname{cc}(\mathcal C_\pi)
+\operatorname{cc}(\mathcal G_\pi)-|\pi|}
p^{|\pi|}.
\]

To interpret this exponent combinatorially, define the bipartite multigraph
\(\mathcal T_\pi\) as follows. Its left vertices are the connected components
of \(\mathcal C_\pi\), its right vertices are the connected components of
\(\mathcal G_\pi\), and each block of \(\pi\) forms an edge joining the two
components that contain it. Thus, \(\mathcal T_\pi\) is a loopless bipartite multigraph, although it may
have parallel edges.
The graph \(\mathcal T_\pi\) is connected, with
\(|\pi|\) edges and
\[
\operatorname{cc}(\mathcal C_\pi)
+\operatorname{cc}(\mathcal G_\pi)
\]
vertices. Its cycle rank is therefore
\[
g_\pi
=
|\pi|
-\operatorname{cc}(\mathcal C_\pi)
-\operatorname{cc}(\mathcal G_\pi)
+1.
\]

Consequently,
\[
    |\Cont(\pi)|
    \le d^{1-g_\pi}p^{|\pi|}.
\]
\paragraph{Canonical walks and the weighted count}
It remains to sum these bounds over all admissible partitions \(\pi\).
The factor \(d^{-g_\pi}\) penalizes cyclic structure, but exploiting this
penalty requires counting partitions while keeping track of their cycle rank.
We do this by encoding each partition as a canonical walk whose support has
cycle rank \(g_\pi\).

Read the \(2\ell\) trace positions in cyclic order:
\[
(\gamma_1,i_1),(\gamma_1,j_1),
(\gamma_2,i_2),(\gamma_2,j_2),
\ldots,
(\gamma_\ell,i_\ell),(\gamma_\ell,j_\ell).
\]
Each position belongs to a block of \(\pi\), and each block is an edge of
\(\mathcal T_\pi\). For every \(t\), the blocks containing
\((\gamma_t,i_t)\) and \((\gamma_t,j_t)\) lie in the same connected component
of \(\mathcal C_\pi\), so the corresponding edges of \(\mathcal T_\pi\) share
a left endpoint. Similarly, the blocks containing \((\gamma_t,j_t)\) and
\((\gamma_{t+1},i_{t+1})\) lie in the same connected component of
\(\mathcal G_\pi\), so their edges share a right endpoint, where \(t+1\) is
interpreted cyclically. Thus, the sequence of corresponding edges forms a
closed walk of length \(2\ell\) in \(\mathcal T_\pi\).

We encode the walk canonically by relabeling its vertices and edges in order
of first appearance. The resulting encoding forgets the original names of the
blocks and connected components, recording only the order in which vertices
and edges are visited and revisited. The \emph{support} of a walk \(W\) is the
multigraph formed by the distinct vertices and edges visited by \(W\). Let
\(v(W)\) and \(e(W)\) denote the numbers of vertices and edges in its support.
Since the support is connected, its cycle rank is
\[
g(W)=e(W)-v(W)+1.
\]

The edge corresponding to a block \(D\) is traversed once for each trace
position in \(D\), and hence exactly \(|D|\) times. Since every block of an
admissible partition has even size, every edge in the support is traversed a
positive even number of times. Conversely, the canonical walk recovers the
partition of the trace positions: two positions lie in the same block if and
only if the corresponding steps traverse the same support edge. Thus, the map
from admissible partitions to canonical walks is injective. Moreover, the
support of the walk is \(\mathcal T_\pi\), and therefore
\[
e(W)=|\pi|,
\qquad
g(W)=g_\pi.
\]

Let \(\mathcal W_\ell\) be the class of all canonical walks of length
\(2\ell\) with loopless support such that every support edge is traversed a
positive even number of times. Every admissible partition gives a distinct
walk in \(\mathcal W_\ell\), although not every walk in \(\mathcal W_\ell\)
need arise from an admissible partition. Since the weights
\(d^{-g(W)}p^{e(W)}\) are nonnegative, enlarging the sum gives
\[
\mathbb E\operatorname{tr}(G^\ell)
\le
d\sum_{W\in\mathcal W_\ell}
d^{-g(W)}p^{e(W)}.
\]

The problem is now to bound a weighted count of canonical walks, organized by
the cycle rank of their support. When \(g(W)=0\), the support is a tree, and
the corresponding walks can be counted directly. When \(g(W)\ge1\), the
cycles introduce additional possibilities in the traversal. Each independent
cycle, however, contributes a factor \(d^{-1}\) to the weight. We will show
that, when \(\ell=O(\log d)\), this penalty compensates for the additional
possibilities, and hence
\[
\sum_{W\in\mathcal W_\ell} d^{-g(W)}p^{e(W)}
\le (Cp)^\ell.
\]

First suppose that \(g(W)=0\), so the support is a tree. Root the tree at the
initial vertex of the walk, and record whether each step increases or decreases
the distance from the root. The resulting sequence is a Dyck word: the distance
is initially zero, remains nonnegative, and returns to zero after \(2\ell\)
steps. Once the Dyck word is fixed, every step toward the root is determined by
the current vertex. At each step away from the root, the walk either traverses
an existing child edge or discovers a new one. Every support edge is discovered
exactly once in this way, so the weight \(p^{e(W)}\) contributes one factor of
\(p\) for each discovery. Using \(p\ge\ell\), the total weighted contribution
of the tree-supported walks is at most
\[
(Cp)^\ell
\]
for a universal constant \(C\).

Now suppose that \(g(W)\ge1\). Fix a spanning tree of the support. Since the
support has \(v(W)\) vertices and \(e(W)\) edges, exactly
\[
e(W)-(v(W)-1)=g(W)
\]
edges lie outside the spanning tree. Thus, a cyclic support consists of a tree
together with \(g(W)\) additional edges. The tree part is handled by the
argument above, while the additional choices associated with the non-tree
edges are controlled by the factor \(d^{-g(W)}\).

The weighted walk-count argument in Section~\ref{sec:walk-count} makes this precise. It shows that, provided
\(\ell\le c\log d\) and \(p\ge\ell\),
\[
\sum_{W\in\mathcal W_\ell}
d^{-g(W)}p^{e(W)}
\le (Cp)^\ell
\]
for universal constants \(c,C>0\). Consequently,
\[
\mathbb E\operatorname{tr}(G^\ell)
\le
d(Cp)^\ell.
\]
Since \(G\) is symmetric and \(\ell\) is even,
\[
\mathbb E\|E\|_{\mathrm{op}}
=
\frac1s\mathbb E\|G\|_{\mathrm{op}}
\le
\frac1s
\left(\mathbb E\operatorname{tr}(G^\ell)\right)^{1/\ell}
\le
C d^{1/\ell}\frac{p}{s}.
\]
Choosing \(\ell=\Theta(\log d)\) makes \(d^{1/\ell}\) bounded by a universal
constant, while \(p/s = O(\varepsilon)\). Hence
\[
\mathbb E\|E\|_{\mathrm{op}} = O(\varepsilon)
\]
as desired.

\subsection{Spectral concentration}

The concentration argument depends on the actual mean spectral error
and does not use the parameter regime in
\eqref{eq:technical-regime}.

\begin{proposition}[Mean-to-moments concentration]
\label{prop:mean-to-moments}
Let $n\ge d\ge2$, let $s,B$ be positive integers, set $m=sB$, and let
$\Phi\sim\operatorname{SparseStack}(s,B)$. For a fixed
$U\in\R^{n\times d}$ with $U^\top U=I_d$, set
\[
    E=U^\top\Phi^\top\Phi U-I_d,
    \qquad
    R=\lVert E\rVert_{\op},
    \qquad
    \mu=\E R.
\]
There exists a universal constant $C_{\mathrm{conc}}$ such that, for every
real $q\ge4\log(2d^2)$,
\begin{equation}\label{eq:mean-to-moments}
    \lVert R\rVert_q
    \le
    C_{\mathrm{conc}}\left(
        \mu+\sqrt{\frac qm}+\frac qs
    \right).
\end{equation}
Moreover, for every $t\ge0$,
\begin{equation}\label{eq:mean-to-tail}
    \Pp\!\left\{
        R>
        C_{\mathrm{conc}}\left(
            \mu
            +\sqrt{\frac{\log(2d)+t}{m}}
            +\frac{\log(2d)+t}{s}
        \right)
    \right\}
    \le e^{-t}.
\end{equation}
\end{proposition}

We prove the moment bound by treating the lower and upper distortions
separately. The tail bound then follows from the moment bound by Markov's
inequality.

For \(a\in\mathbb R\), write \((a)_+=\max\{a,0\}\), and define
\[
Y:=(-\lambda_{\min}(E))_+,
\qquad
Z:=(\lambda_{\max}(E))_+.
\]
Since \(E\) is symmetric,
\[
\|E\|_{\op}=\max\{Y,Z\}.
\]

These quantities have a direct geometric interpretation. For every unit vector
\(x\in\mathbb R^d\),
\[
x^\top E x=\|\Phi Ux\|_2^2-1.
\]
Consequently,
\[
Y=\max_{\|x\|_2=1}\bigl(1-\|\Phi Ux\|_2^2\bigr)_+,
\qquad
Z=\max_{\|x\|_2=1}\bigl(\|\Phi Ux\|_2^2-1\bigr)_+.
\]
Thus, \(Y\) is the maximum decrease in squared norm over unit vectors in \(V\),
while \(Z\) is the maximum increase. We call \(Y\) and \(Z\) the
\emph{lower distortion} and \emph{upper distortion}, respectively.

We prove separate moment bounds for \(Y\) and \(Z\) using different
arguments. \space
A layer may increase the squared norm by an arbitrarily large amount, but it cannot contribute more than \(1/s\) to its decrease.
For the lower distortion, this one-sided boundedness permits an application of Bousquet's
inequality~\cite{Bousquet2002}. No comparable upper bound is available for the
upper distortion, because many coordinates may collide in one bucket and
reinforce one another. We therefore control the lower distortion layer by
layer, while for the upper distortion we resample one individual random choice at a time.

\paragraph{Lower distortion.}
Let
\[
    T_\gamma=U^\top S_\gamma^\top S_\gamma U.
\]
Then
\[
    I_d+E=\frac1s\sum_{\gamma=1}^s T_\gamma,
    \qquad
    T_\gamma\succeq0,
    \qquad
    \E T_\gamma=I_d.
\]
For a unit vector \(v\), define the centered layerwise lower deviation
\[
    L_{\gamma,v}
    :=
    \frac{1-v^\top T_\gamma v}{s}.
\]
The lower distortion can then be written as
\[
    Y
    =
    \max\left\{
        0,\,
        \sup_{\lVert v\rVert_2=1}
        \sum_{\gamma=1}^s L_{\gamma,v}
    \right\}.
\]
For each \(v\), the variables \(L_{\gamma,v}\) are independent
and centered. Moreover, since \(T_\gamma\succeq0\),
\[
    L_{\gamma,v}\le\frac1s.
\]
This one-sided bound permits the use of Bousquet's inequality.

A direct calculation using the independence of the signs and the
collision probability \(1/B\) gives the variance bound
\[
    \sup_{\lVert v\rVert_2=1}
    \sum_{\gamma=1}^s
    \E L_{\gamma,v}^2
    =O\!\left(\frac1m\right).
\]
Since the variables \(L_{\gamma,v}\) are independent and centered, satisfy
\(L_{\gamma,v}\le 1/s\), and obey the variance bound above, Bousquet's inequality~\cite{Bousquet2002} gives
\[
    \lVert Y\rVert_q
    \le
    \E Y
    +O\!\left(
        \sqrt{\frac qm}
        +\sqrt{\frac{q\,\E Y}{s}}
        +\frac qs
    \right).
\]
Since $\E Y\le\mu$, the inequality
\[
    2\sqrt{\frac{\mu q}{s}}
    \le \mu+\frac qs
\]
gives
\[
    \lVert Y\rVert_q
    =
    O\!\left(
        \mu+\sqrt{\frac qm}+\frac qs
    \right).
\]

\paragraph{Upper distortion and coordinate resampling.}
For the upper distortion, there is no useful deterministic upper bound on the
contribution of a single layer. We instead measure the sensitivity of \(Z\) to
the individual random choices. Index these choices by
\[
    \alpha=(\gamma,i),
    \qquad
    X_\alpha=(h_\gamma(i),\sigma_{\gamma i}).
\]
Let \(X_\alpha'\) be an independent copy of \(X_\alpha\), and let
\(E^{(\alpha)}\) and \(Z^{(\alpha)}\) denote the error matrix and upper
distortion obtained by replacing \(X_\alpha\) with \(X_\alpha'\). Define
\[
    \mathcal V^+
    =
    \sum_\alpha
    \E_\alpha'
    \left[(Z-Z^{(\alpha)})_+^2\right],
\]
where \(\E_\alpha'\) averages over \(X_\alpha'\), with all the original random
choices held fixed. The quantity \(\mathcal V^+\) is a one-sided variance
proxy: it measures the total squared decrease in the current upper distortion
under individual resampling. Polynomial Efron--Stein
inequality~\cite{BoucheronEtAl2005} gives
\[
    \bigl\|(Z-\E Z)_+\bigr\|_q
    =
    O\!\left(
        \sqrt{q\,\lVert\mathcal V^+\rVert_{q/2}}
    \right).
\]
It therefore remains to bound \(\lVert\mathcal V^+\rVert_{q/2}\).

Suppose that \(Z>0\), and fix a unit eigenvector \(x\) corresponding to
\(\lambda_{\max}(E)\). Then
\[
x^\top Ex=\lambda_{\max}(E)=Z.
\]
By the Rayleigh--Ritz variational principle,
\[
Z^{(\alpha)}
=
\bigl(\lambda_{\max}(E^{(\alpha)})\bigr)_+
\ge x^\top E^{(\alpha)}x.
\]
Therefore,
\[
(Z-Z^{(\alpha)})_+
\le
\bigl(x^\top(E-E^{(\alpha)})x\bigr)_+.
\]
Thus, the decrease in upper distortion under resampling is controlled by the
change in the quadratic form along the fixed top eigenvector \(x\) of the
original matrix.

Recall that \(u_i^\top\) is the \(i\)th row of \(U\).
For a unit vector \(v\in\mathbb{R}^d\), define
\[
    t_{\gamma i}(v)
    :=
    \frac{\sigma_{\gamma i}\langle u_i,v\rangle}{\sqrt{s}},
    \qquad
    A_\beta(v)
    :=
    \sum_{i:\,h_\gamma(i)=b}t_{\gamma i}(v).
\]
Therefore, $t_{\gamma i}(v)$ represents a signed contribution and $A_\beta(v)$ represents a bucket load.
We use $t_{\gamma i}$ and $A_\beta$ as shorthand notations for $t_{\gamma i}(x)$ and $A_\beta(x)$ respectively, where $x$ is fixed as described above. Thus,
\(
A_{(\gamma,b)}=(\Phi Ux)_{(\gamma,b)}.
\)

Consider \(\alpha=(\gamma,i)\), and write
\[
h_i:=h_\gamma(i),
\qquad
h_i':=h_\gamma'(i)
\]
for the original and resampled buckets of coordinate \(i\). Similarly, let
\(t_i\) and \(t_i'\) denote its signed contributions before and after
resampling. For each bucket \(b\), define the leave-one-out load
\[
A_{\gamma,b}^{(-i)}
:=
A_{(\gamma,b)}-t_i\mathbf 1\{h_i=b\}.
\]
Thus, \(A_{\gamma,h_i}^{(-i)}\) and \(A_{\gamma,h_i'}^{(-i)}\) are the
leave-one-out loads in the original and resampled buckets, respectively.
 Since resampling changes
only this coordinate and \((t_i')^2=t_i^2\),
\[
Z-Z^{(\alpha)}
\le
x^\top(E-E^{(\alpha)})x
=
2\left(
t_iA_{\gamma,h_i}^{(-i)}
-
t_i'A_{\gamma,h_i'}^{(-i)}
\right).
\]
Hence, writing \((a)_-:=(-a)_+\),
\[
\begin{aligned}
\left(Z-Z^{(\alpha)}\right)_+^2
&\le
4\left(
t_iA_{\gamma, h_i}^{(-i)}
-
t_i'A_{\gamma, h_i'}^{(-i)}
\right)_+^2 \\
&\le
8\left(t_iA_{\gamma, h_i}^{(-i)}\right)_+^2
+
8\left(t_i'A_{\gamma, h_i'}^{(-i)}\right)_-^2.
\end{aligned}
\]
Therefore,
\[
\mathcal V^+
\le
8\underbrace{
\sum_{\alpha}
\left(t_iA_{\gamma, h_i}^{(-i)}\right)_+^2
}_{J_{\mathrm{original}}}
+
\,\,
8
\underbrace{
\sum_{\alpha}\E_\alpha'
\left[
\left(t_i'A_{\gamma, h_i'}^{(-i)}\right)_-^2
\right]
}_{J_{\mathrm{resampled}}}.
\]

The two terms correspond to the two interactions through which resampling can
decrease the quadratic form in the fixed direction \(x\). In the original
bucket, \(t_iA_{\gamma, h_i}^{(-i)}>0\) means that coordinate \(i\) reinforces the
leave-one-out load, so removing it decreases the quadratic form. In the
resampled bucket, \(t_i'A_{\gamma, h_i'}^{(-i)}<0\) means that the reinserted
coordinate cancels part of the leave-one-out load, which also decreases the
quadratic form.

\paragraph{The resampled bucket.}
Fix \(\alpha=(\gamma,i)\) and condition on the current sketch. Then \(x\) and
the leave-one-out bucket loads are fixed, while the resampled bucket is uniform
on \([B]\) and the resampled sign is independent. The negative part in
\(J_{\mathrm{resampled}}\) selects the case in which the reinserted coordinate
partially cancels the load in its resampled bucket. Averaging over the resampled sign
and bucket, and then summing over all coordinates and layers, gives
\[
    J_{\mathrm{resampled}}
    \le
    \frac{1}{m}
    \left(
        \sum_\beta A_\beta(x)^2+1
    \right)
    =
    \frac{Z+2}{m}.
\]
Indeed, since \(x\) is a unit top eigenvector of \(E\) and \(Z>0\),
\[
    \sum_\beta A_\beta(x)^2
    =
    \|\Phi Ux\|_2^2
    =
    1+x^\top Ex
    =
    1+Z.
\]
Thus, uniform resampling averages over all buckets rather than favoring those
with unusually large leave-one-out loads.

\paragraph{The original bucket.}
The original-bucket contribution requires a bucketwise analysis.
For a unit vector \(v\), define
\[
J_\beta(v)
:=
\sum_{i\in I_\beta}
\left(
t_{\gamma i}(v)
\bigl(A_\beta(v)-t_{\gamma i}(v)\bigr)
\right)_+^2,
\qquad \beta=(\gamma,b).
\]
For the maximizing direction \(x\) fixed above,
\[
J_{\mathrm{original}}
=
\sum_\beta J_\beta(x).
\]
The expression inside the positive part is the interaction of coordinate \(i\)
with the other coordinates in its bucket:
\[
t_{\gamma i}(v)
\bigl(A_\beta(v)-t_{\gamma i}(v)\bigr)
=
\sum_{\substack{j\in I_\beta\\j\ne i}}
t_{\gamma i}(v)t_{\gamma j}(v).
\]
Summing these interactions over \(i\in I_\beta\), define
\[
\begin{aligned}
Q_\beta(v)
&:=
\sum_{i\in I_\beta}
t_{\gamma i}(v)
\bigl(A_\beta(v)-t_{\gamma i}(v)\bigr) \\
&=
A_\beta(v)^2
-
\sum_{i\in I_\beta}t_{\gamma i}(v)^2 \\
&=
2\sum_{\substack{i,j\in I_\beta\\i<j}}
t_{\gamma i}(v)t_{\gamma j}(v).
\end{aligned}
\]
Thus, \(Q_\beta(v)\) is the signed bucket excess within bucket \(\beta\):
positive values indicate net reinforcement, while negative values indicate net
cancellation.

However, \(Q_\beta(v)\) alone does not control \(J_\beta(v)\).
The quantity \(Q_\beta(v)\) is the signed sum of the coordinatewise
interactions, whereas \(J_\beta(v)\) takes the positive part of each
coordinatewise interaction and squares it before summing. Thus, positive and
negative coordinatewise interactions may cancel in \(Q_\beta(v)\), even when
\(J_\beta(v)\) is large. To control the interactions that may be hidden by
this cancellation, define
\[
\begin{aligned}
K_\beta(v)
&:=
\left(
\sum_{i\in I_\beta}t_{\gamma i}(v)^2
\right)^2
-
\sum_{i\in I_\beta}t_{\gamma i}(v)^4 \\
&=
2\sum_{\substack{i,j\in I_\beta\\i<j}}
t_{\gamma i}(v)^2t_{\gamma j}(v)^2.
\end{aligned}
\]
Unlike \(Q_\beta(v)\), the quantity \(K_\beta(v)\) sums squared pairwise
interactions and therefore admits no sign cancellation.

We can then show that
\[
J_\beta(v)
=
O\!\left(
K_\beta(v)+\psi\bigl(Q_\beta(v)\bigr)
\right)
\]
where
\[
\psi(z)
=
\min\left\{
(z)_+^2,\,
\frac{(z)_+}{s}
\right\}.
\]

Since \(J_{\mathrm{original}}\) is evaluated at the maximizing direction \(x\),
we introduce the global quantities
\[
\mathcal K
:=
\sup_{\|v\|_2=1}
\sum_\beta K_\beta(v),
\qquad
\mathcal Q
:=
\sup_{\|v\|_2=1}
\sum_\beta\psi\bigl(Q_\beta(v)\bigr).
\]
The bucketwise bound then implies
\[
J_{\mathrm{original}}
=
O\!\left(\mathcal K+\mathcal Q\right).
\]
Combining this with the resampled-bucket bound, we obtain
\[
\mathcal V^+
=
O\!\left(
\mathcal K
+
\mathcal Q
+
\frac{1+Z}{m}
\right).
\]

It remains to bound \(\mathcal K\) and \(\mathcal Q\). The quantity
\(\mathcal K\) depends only on the hash functions, since the signs disappear
after squaring. We express it in terms of a sum of independent
positive semidefinite random matrices and apply matrix Bernstein to obtain
\[
\|\mathcal K\|_q
=
O\!\left(
\frac1m+\frac{q}{s^2}
\right).
\]

To control \(\mathcal Q\), we condition on the hash functions. The quantities
associated with distinct buckets then depend on disjoint sets of independent
signs. The bound
\[
\psi(z)\le z^2
\]
controls the conditional expectation in terms of \(\mathcal K\), while the
\(O(1/s)\)-Lipschitz property of \(\psi\), together with Rademacher
contraction, controls the fluctuations around that expectation. Finally, the
identity
\[
\sum_\beta Q_\beta(v)
=
\|\Phi Uv\|_2^2-1
=
v^\top Ev
\]
relates the resulting Rademacher process to the spectral error. These
ingredients yield
\[
\|\mathcal Q\|_q
=
O\!\left(
\|\mathcal K\|_q
+
\frac1s
\bigl\|\|E\|_{\mathrm{op}}\bigr\|_q
\right).
\]

\paragraph{Closing the moment bound.}
Let
\[
    r_q
    :=
    \bigl\lVert\lVert E\rVert_{\op}\bigr\rVert_q.
\]
Combining the variance-proxy bound with the bounds for
\(\mathcal K\) and \(\mathcal Q\), and using \(m\ge s\), gives
\[
    \lVert\mathcal V^+\rVert_{q/2}
    =
    O\!\left(
        \frac1m+\frac q{s^2}+\frac{r_q}{s}
    \right).
\]
The \(L_q\) Efron--Stein inequality and $\E Z\le\mu$ yield
\[
    \lVert Z\rVert_q
    =
    O\!\left(
        \mu
        +\sqrt{\frac qm}
        +\frac qs
        +\sqrt{\frac{q r_q}{s}}
    \right).
\]

Since
\[
    \lVert E\rVert_{\op}
    =
    \max\{Y,Z\}
    \le Y+Z,
\]
combining this bound with the lower-distortion bound gives
\[
    r_q
    =
    O\!\left(
        \mu
        +\sqrt{\frac qm}
        +\frac qs
        +\sqrt{\frac{q r_q}{s}}
    \right).
\]
The last term is self-referential, but Young's inequality gives
\[
    C\sqrt{\frac{q r_q}{s}}
    \le
    \frac12 r_q+C'\frac qs.
\]
Absorbing \(r_q/2\) into the left-hand side, we obtain
\[
    r_q
    =
    O\!\left(
        \mu+\sqrt{\frac qm}+\frac qs
    \right).
\]

This proves the moment bound in
Proposition~\ref{prop:mean-to-moments}. To obtain the tail bound, choose
\[
    q=4\log(2d^2)+\log \frac{1}{\delta} =\Theta\!\left(\log\frac d\delta\right).
\]
Then \(q\ge 4\log(2d^2)\), so the moment bound applies. Applying Markov's inequality to the nonnegative random variable \(\lVert E\rVert_{\op}^q\), and using
\(r_q^q=\E\lVert E\rVert_{\op}^q\), we have
\[
    \Pp\!\left\{\lVert E\rVert_{\op}>e \cdot r_q\right\}
    =
    \Pp\!\left\{
        \lVert E\rVert_{\op}^q>(e \cdot r_q)^q
    \right\}
    \le
    \frac{\E\lVert E\rVert_{\op}^q}{(e \cdot r_q)^q}
    =
    e^{-q} \le \delta.
\]
Moreover, our choices of \(m\) and \(s\) give
\[
    \sqrt{\frac qm}=O(\eps),
    \qquad
    \frac qs=O(\eps).
\]
Therefore
\[
    \Pp\!\left\{\lVert E\rVert_{\op}> O(\eps)\right\} \le \delta.
\]
Scaling $\eps$ by a constant factor gives the desired tail bound.

\subsection{Proof organization}

Sections~\ref{sec:mean-error}--\ref{sec:walk-count} prove the mean bound using Eulerian Gram sums and a weighted count of canonical walks.
Section~\ref{sec:shrinkage} proves the lower-distortion bound.
Sections~\ref{sec:expansion-sensitivity}--\ref{sec:expansion-recursion}
control upper distortion and prove the operator-norm moment and tail bounds.
Finally, Section~\ref{sec:main-proof} chooses the parameters and proves
Theorem~\ref{thm:main}. The appendices contain the full Eulerian contraction
proof, the bounded-dimensional trace proof, the algebraic bucket inequality,
and the standard probability tools.

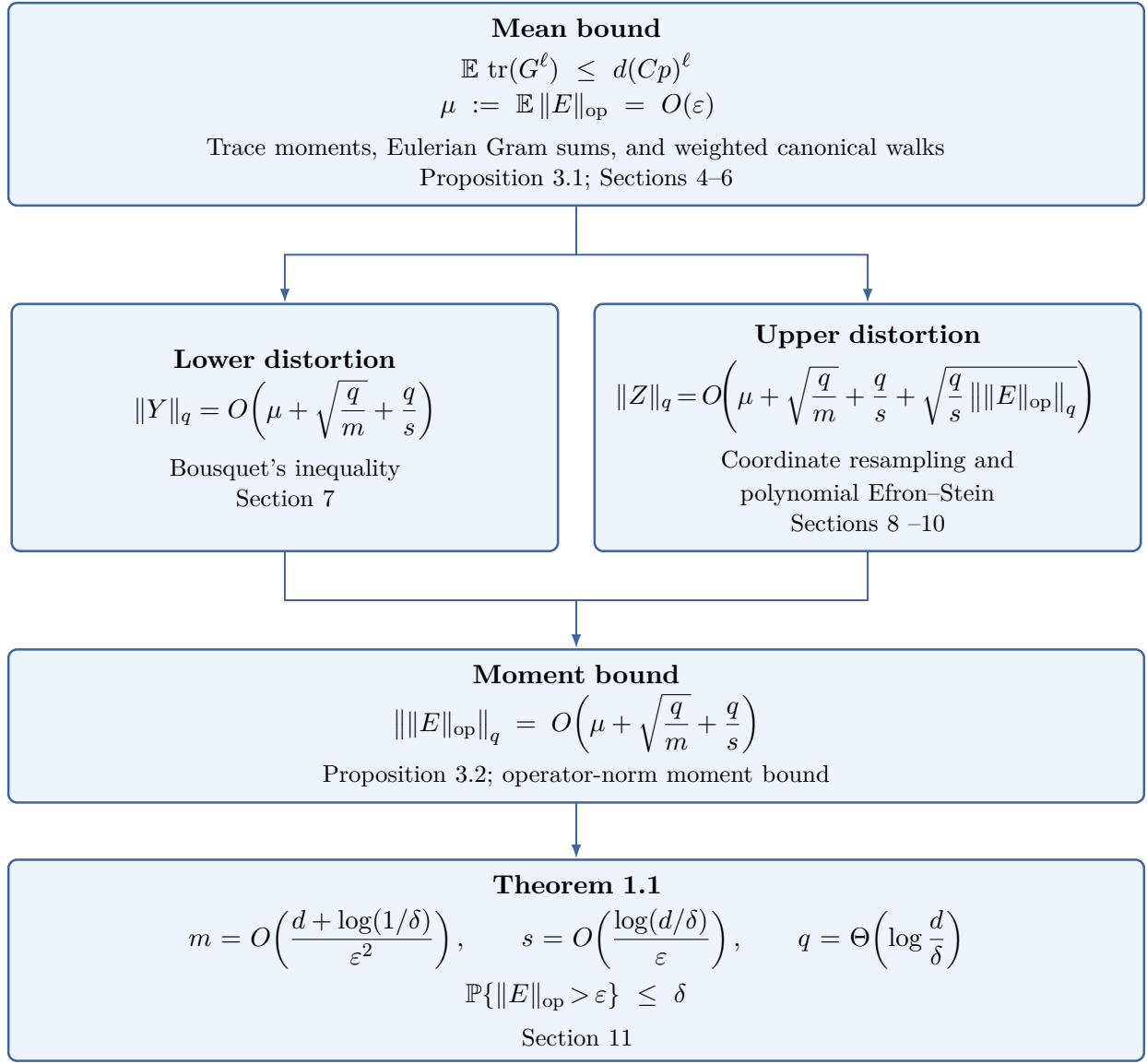
\begin{figure}[H]
  \centering
  \begin{tikzpicture}[
    scale=0.985,
    transform shape,
    font=\normalsize,
    roadmap box/.style={
      draw=black!65,
      fill=layerfill,
      rounded corners=3.5pt,
      line width=0.60pt,
      align=center,
      inner xsep=7.5pt,
      inner ysep=6.5pt
    },
    roadmap result/.style={
      roadmap box,
      draw=accent,
      fill=zoomfill,
      line width=0.95pt
    },
    roadmap line/.style={
      draw=accent,
      line width=0.75pt
    },
    roadmap arrow/.style={
      roadmap line,
      -{Latex[length=2.0mm,width=1.45mm]}
    }
  ]
    \thickmuskip=3mu plus 1mu minus 1mu\relax

    \node[
      roadmap result,
      text width=0.97\linewidth
    ] (mean) {
      \textbf{Mean bound}\\[3pt]
      $\displaystyle
        \mathbb{E}\,\operatorname{tr}(G^\ell)
        \le d(Cp)^\ell
      $\\[1pt]
      $\displaystyle
        \mu:=\mathbb{E}\,\lVert E\rVert_{\mathrm{op}}
        =
        O(\varepsilon)
      $\\[4pt]
      {\small
        Trace moments, Eulerian Gram sums,
        and weighted canonical walks
      }\\[-1pt]
      {\small
        Proposition~\ref{prop:mean-error};
        Sections~\ref{sec:mean-error}--\ref{sec:walk-count}
      }
    };

    \coordinate (split) at ([yshift=-7mm]mean.south);

    \node[
      roadmap result,
      text width=0.45\linewidth,
      minimum height=36mm,
      anchor=north east
    ] (lower) at ([xshift=-2.5mm,yshift=-7mm]split) {
      \textbf{Lower distortion}\\[4pt]
      $\displaystyle
        \lVert Y\rVert_q
        =
        O\!\left(
          \mu
          +
          \sqrt{\frac{q}{m}}
          +
          \frac{q}{s}
        \right)
      $\\[7pt]
      {\small Bousquet's inequality}\\[-1pt]
      {\small Section~\ref{sec:shrinkage}}
    };

    \node[
      roadmap result,
      text width=0.45\linewidth,
      minimum height=36mm,
      anchor=north west
    ] (upper) at ([xshift=2.5mm,yshift=-7mm]split) {
      \textbf{Upper distortion}\\[3pt]
      $\displaystyle
        \begin{aligned}
          \lVert Z\rVert_q
          =
          O\!\Bigg(&
            \mu
            +
            \sqrt{\frac{q}{m}}
            +
            \frac{q}{s}
            +
            \sqrt{
              \frac{q}{s}\,
              \bigl\lVert
                \lVert E\rVert_{\mathrm{op}}
              \bigr\rVert_q
            }
          \Bigg)
        \end{aligned}
      $\\[5pt]
      {\small
        Coordinate resampling and polynomial Efron--Stein
      }\\[-1pt]
      {\small
        Sections~\ref{sec:expansion-sensitivity}
        --\ref{sec:expansion-recursion}
      }
    };

    \coordinate (join) at
      ([yshift=-7mm]$(lower.south)!0.5!(upper.south)$);

    \node[
      roadmap result,
      text width=0.97\linewidth,
      anchor=north
    ] (concentration) at ([yshift=-7mm]join) {
      \textbf{Moment bound}\\[3pt]
      $\displaystyle
        \bigl\lVert
          \lVert E\rVert_{\mathrm{op}}
        \bigr\rVert_q
        =
        O\!\left(
          \mu
          +
          \sqrt{\frac{q}{m}}
          +
          \frac{q}{s}
        \right)
      $\\[4pt]
      {\small
        Proposition~\ref{prop:mean-to-moments};
        operator-norm moment bound
      }
    };

    \node[
      roadmap result,
      text width=0.97\linewidth,
      below=8mm of concentration
    ] (theorem) {
      \textbf{Theorem~\ref{thm:main}}\\[5pt]
      $\displaystyle
        m
        =
        O\!\left(
          \frac{d+\log(1/\delta)}{\varepsilon^2}
        \right),
        \qquad
        s
        =
        O\!\left(
          \frac{\log(d/\delta)}{\varepsilon}
        \right),
        \qquad
        q=\Theta\!\left(\log\frac d\delta\right)
      $\\[5pt]
      $\displaystyle
        \Pp\!\left\{
          \lVert E\rVert_{\mathrm{op}}>\varepsilon
        \right\}
        \le \delta
      $\\[4pt]
      {
      \small
      Section~\ref{sec:main-proof}
      }
    };

    \draw[roadmap line]
      (mean.south) -- (split);

    \draw[roadmap arrow]
      (split) -| (lower.north);

    \draw[roadmap arrow]
      (split) -| (upper.north);

    \draw[roadmap line]
      (lower.south) |- (join);

    \draw[roadmap line]
      (upper.south) |- (join);

    \draw[roadmap arrow]
      (join) -- (concentration.north);

    \draw[roadmap arrow]
      (concentration.south) -- (theorem.north);

  \end{tikzpicture}

  \caption{Proof roadmap for Theorem~\ref{thm:main}.}
  \label{fig:proof-roadmap}
\end{figure}

\section{Mean spectral error}\label{sec:mean-error}

This section proves Proposition~\ref{prop:mean-error}.
Let
\begin{equation*}
    G=sE=U^\top(S^\top S-sI_n)U.
\end{equation*}
The proof of the mean bound proceeds by expanding the trace
moment \(\E\tr(G^\ell)\). The resulting sum is controlled using two lemmas, proved separately: one bounds injective sums of Gram products indexed by Eulerian multigraphs, and the other bounds a weighted count of canonical walks. Their proofs are deferred to Sections~\ref{sec:gram} and~\ref{sec:walk-count}. The moment order used here is denoted by $\ell = O(\log d)$ and depends only on $d$.

\subsection{Trace expansion}

For $i\ne j$, the $(i,j)$ entry of $S^\top S$ is
\[
    (S^\top S)_{ij}
    =\sum_{\gamma=1}^s
    \one\{h_\gamma(i)=h_\gamma(j)\}\sigma_{\gamma i}\sigma_{\gamma j}.
\]
The diagonal entries of $S^\top S-sI_n$ vanish. Consequently, the expansion of $G$ contains only terms indexed by $i\ne j$. Since the $i$th row of $U$ is $u_i^\top$, we have
\begin{equation}\label{eq:G-rank-one}
    G=
    \sum_{\gamma=1}^s\sum_{i\ne j}
    \one\{h_\gamma(i)=h_\gamma(j)\}
    \sigma_{\gamma i}\sigma_{\gamma j}
    u_i u_j^\top.
\end{equation}
Expanding the $\ell$-fold product $G^\ell$ from \eqref{eq:G-rank-one}, choosing one summand from each factor gives a term
\[
    \prod_{t=1}^{\ell}
    \sigma_{\gamma_t i_t}\sigma_{\gamma_t j_t}\,
    u_{i_t}u_{j_t}^\top,
\]
indexed by the tuple of choices
\[
    \tau=((\gamma_t,i_t,j_t))_{t=1}^\ell,
    \qquad i_t\ne j_t.
\]
Let $\tT_\ell$ denote the set of all such tuples.

For rank-one matrices,
\[
    (u_i u_j^\top)(u_k u_v^\top)
    =\langle u_j,u_k\rangle u_i u_v^\top.
\]
Applying this identity repeatedly and taking the trace gives, for each tuple $\tau=((\gamma_t,i_t,j_t))_{t=1}^\ell\in\tT_\ell$,
\begin{equation}\label{eq:rank-one-trace}
    \tr\!\left[(u_{i_1}u_{j_1}^\top)\cdots(u_{i_\ell}u_{j_\ell}^\top)\right]
    =\prod_{t=1}^\ell\langle u_{j_t},u_{i_{t+1}}\rangle,
    \qquad i_{\ell+1}=i_1.
\end{equation}
Combining this with the sign factor from each term of $G^\ell$, and writing $\E_\sigma,\E_h$ for expectation over the signs and hashes respectively, define for $\tau\in\tT_\ell$
\begin{align*}
    \Sign(\tau)
    &:=\E_\sigma\prod_{t=1}^\ell
       \sigma_{\gamma_t i_t}\sigma_{\gamma_t j_t}
    &&\text{(sign factor)},\\
    \Hash(\tau)
    &:=\E_h\prod_{t=1}^\ell
       \one\{h_{\gamma_t}(i_t)=h_{\gamma_t}(j_t)\}
    &&\text{(hash factor)},\\
    \Gram(\tau)
    &:=\prod_{t=1}^\ell\langle u_{j_t},u_{i_{t+1}}\rangle
    &&\text{(the product in \eqref{eq:rank-one-trace})}.
\end{align*}
The signs and hash functions are independent, so taking the expectation of $\tr G^\ell$ term by term gives
\begin{equation}\label{eq:tuple-sum}
    \E\tr G^\ell
    =\sum_{\tau\in\tT_\ell}
      \Sign(\tau)\Hash(\tau)\Gram(\tau).
\end{equation}

\subsection{Partitions and the two Eulerian multigraphs}

Our goal is to reorganize the sum over terms $\tau=((\gamma_t,i_t,j_t))_{t=1}^\ell$ in \eqref{eq:tuple-sum} into a sum over a much smaller set of objects.
This organization compresses many differently labeled tuples into a single combinatorial object and exposes the underlying collision and interaction structure.

Every term $\tau$ can be encoded by two pieces of data: a partition of $2\ell$ trace positions recording which positions share a layer--coordinate pair, and assignments of a layer and a coordinate to each block, subject to an injectivity constraint described below.

Consider the $2\ell$ trace positions
\[
    \cO_\ell=\{L_1,R_1,\ldots,L_\ell,R_\ell\}.
\]
A partition of $\cO_\ell$ groups these positions into disjoint nonempty blocks. For example, when $\ell=2$, $\pi= \{ \{L_1,L_2\}, \{R_1,R_2\}\}$ is a partition of $\cO_\ell$ consisting of two blocks.

A term $\tau$ determines a partition $\pi(\tau)$ of $\cO_\ell$: position $L_t$ corresponds to $(\gamma_t,i_t)$, position $R_t$ corresponds to $(\gamma_t,j_t)$, and two positions lie in the same block of $\pi(\tau)$ exactly when they are assigned the same pair. For a position $z\in\cO_\ell$, write $[z]_\pi$ for the block of $\pi$ containing $z$.
Because $i_t\ne j_t$ for every $t$, every partition arising this way keeps $L_t$ and $R_t$ in different blocks:
\begin{equation}\label{eq:no-diagonal-block}
    [L_t]_\pi\ne[R_t]_\pi
    \qquad (t\in[\ell]).
\end{equation}

Conversely, $\tau$ can be recovered from $\pi(\tau)$ together with the layer and coordinate assigned to each block: these are the layer and coordinate assignments defined below, together with the required injectivity constraint, once the collision and Gram multigraphs have been defined.

Below we define two multigraphs, the collision multigraph and the Gram multigraph, using only the partition. They let us bound the sum of the hash and Gram factors over all trace tuples $\tau$ inducing that partition.
For a graph $H$, we use $\cc(H)$ to denote the number of connected components of $H$.

\begin{definition}[Collision and Gram multigraphs]\label{def:AB-graphs}
For each $t\in[\ell]$:
\begin{itemize}
    \item the collision multigraph $\cC_\pi$ has one edge $\{[L_t]_\pi,[R_t]_\pi\}$;
    \item the Gram multigraph $\cG_\pi$ has one edge $\{[R_t]_\pi,[L_{t+1}]_\pi\}$, where $L_{\ell+1}=L_1$.
\end{itemize}
Parallel edges and loops are allowed, but Equation~\eqref{eq:no-diagonal-block} says that $\cC_\pi$ has no loops.
\end{definition}

For every block $D\in\pi$, every trace position $z\in D$ is an endpoint of exactly one edge of $\cC_\pi$ and exactly one edge of $\cG_\pi$: if $z=L_t$ or $z=R_t$, it is an endpoint of the edge $\{[L_t]_\pi,[R_t]_\pi\}$ of $\cC_\pi$ and of an edge of $\cG_\pi$ involving $L_t$ or $R_t$. Summing over the $|D|$ trace positions in $D$ gives
\begin{equation}\label{eq:equal-degrees}
    \deg_{\cC_\pi}(D)=\deg_{\cG_\pi}(D)=|D|.
\end{equation}
Let $\cP_\ell$ be the set of partitions satisfying \eqref{eq:no-diagonal-block} and whose blocks all have even size. We call them \emph{admissible partitions}. If $\pi\in\cP_\ell$, both $\cC_\pi$ and $\cG_\pi$ are Eulerian and have no isolated vertices. Both graphs have vertex set $\pi$, so
\begin{equation*}
    |V(\cC_\pi)|=|V(\cG_\pi)|=|\pi|.
\end{equation*}
Since $\cC_\pi$ is loopless and has no isolated vertices, every connected component contains at least two vertices, and therefore
\begin{equation}\label{eq:component-block-count}
    2\cc(\cC_\pi)\le |\pi|.
\end{equation}

Fix a partition $\pi$ satisfying \eqref{eq:no-diagonal-block}. Let $\cL_\pi$ be the set of layer assignments $\lambda:\pi\to[s]$ such that
\begin{equation}\label{eq:lambda-constraint}
    \lambda([L_t]_\pi)=\lambda([R_t]_\pi)
    \qquad (t\in[\ell]).
\end{equation}
Equivalently, $\lambda$ is constant on every connected component of $\cC_\pi$, so
\begin{equation}\label{eq:lambda-count}
    |\cL_\pi|=s^{\cc(\cC_\pi)}.
\end{equation}
Note that different components of $\cC_\pi$ may correspond to the same layer.

For $\lambda\in\cL_\pi$, let $\cI_{\pi,\lambda}$ be the set of coordinate assignments $a:\pi\to[n]$ satisfying
\begin{equation}\label{eq:row-injective}
    D\ne D'
    \text{ and }
    \lambda(D)=\lambda(D')
    \quad\Longrightarrow\quad
    a(D)\ne a(D').
\end{equation}
This implies that the ordered pairs $(\lambda(D),a(D))$, for $D\in\pi$, are distinct. Coordinate labels may repeat across different layers.

\begin{claim}\label{clm:exact-partition-parametrization}
Fix a partition $\pi$ satisfying \eqref{eq:no-diagonal-block}. The map that assigns to each block its common layer and coordinate is a bijection
\[
    \{\tau\in\tT_\ell:\tau\text{ induces the partition }\pi\}
    \longleftrightarrow
    \bigcup_{\lambda\in\cL_\pi}
    \{\lambda\}\times\cI_{\pi,\lambda}.
\]
Under the inverse map, $(\lambda,a)$ reconstructs the tuple by
\begin{equation}\label{eq:tuple-reconstruction}
    \gamma_t=\lambda([L_t]_\pi)=\lambda([R_t]_\pi),
    \qquad
    i_t=a([L_t]_\pi),
    \qquad
    j_t=a([R_t]_\pi).
\end{equation}
\end{claim}

\begin{proof}
Given a tuple inducing $\pi$, assign to each block its common layer and coordinate. The resulting layer assignment belongs to $\cL_\pi$ by \eqref{eq:lambda-constraint}. If two distinct blocks have the same layer, then their coordinates must be distinct; otherwise their ordered pairs would coincide and the tuple would induce a coarser partition. Hence the coordinate assignment belongs to $\cI_{\pi,\lambda}$.

Conversely, fix $\lambda\in\cL_\pi$ and $a\in\cI_{\pi,\lambda}$, and define the tuple by \eqref{eq:tuple-reconstruction}. The off-diagonal condition follows from \eqref{eq:no-diagonal-block}, \eqref{eq:lambda-constraint}, and \eqref{eq:row-injective}. Positions in the same block receive the same pair $(\gamma,i)$, whereas positions in distinct blocks receive distinct pairs: either their layers differ, or \eqref{eq:row-injective} forces their coordinates to differ. Thus the reconstructed tuple induces $\pi$, and the two constructions are inverse.
\end{proof}
The claim holds for every partition $\pi$ satisfying \eqref{eq:no-diagonal-block}, not only for the admissible partitions in $\cP_\ell$: it only records how a tuple decomposes into a partition and a labeling, without yet using the sign factor. The restriction to $\cP_\ell$ enters below, once we show that the sign factor vanishes unless every block has even size.

For a coordinate assignment $a$, define the partition-level Gram factor
\begin{equation*}
    \Gram_\pi(a)
    :=\prod_{t=1}^\ell
    \bigl\langle u_{a([R_t]_\pi)},u_{a([L_{t+1}]_\pi)}\bigr\rangle.
\end{equation*}
Note that for the tuple corresponding to $(\lambda,a)$, the pairs $(\lambda(D),a(D))$ are distinct across different blocks, so the corresponding Rademacher variables are independent.
 Thus the sign factor satisfies
\begin{equation}\label{eq:sign-expectation}
\Sign(\tau)
=
\E_\sigma
\prod_{t=1}^\ell
\sigma_{\lambda([L_t]_\pi),a([L_t]_\pi)}
\sigma_{\lambda([R_t]_\pi),a([R_t]_\pi)}
=
\prod_{D\in\pi}\E_\sigma\sigma_{\lambda(D),a(D)}^{|D|}
=
\begin{cases}
1,&\pi\in\cP_\ell,\\
0,&\text{otherwise}.
\end{cases}
\end{equation}
For a partition satisfying \eqref{eq:no-diagonal-block}, define the partition-level hash factor
\begin{equation*}
    \Hash_\pi(\lambda,a)
    :=\E_h\prod_{t=1}^\ell
    \one\!\left\{
        h_{\lambda([L_t]_\pi)}(a([L_t]_\pi))
        =h_{\lambda([R_t]_\pi)}(a([R_t]_\pi))
    \right\}.
\end{equation*}
Using Claim~\ref{clm:exact-partition-parametrization} in \eqref{eq:tuple-sum}, followed by \eqref{eq:sign-expectation}, gives the exact identity
\begin{equation}\label{eq:partition-exact}
    \E\tr G^\ell
    =\sum_{\pi\in\cP_\ell}
      \sum_{\lambda\in\cL_\pi}
      \sum_{a\in\cI_{\pi,\lambda}}
      \Hash_\pi(\lambda,a)\Gram_\pi(a).
\end{equation}

\subsection{The contribution of one partition}
In this section, we bound the contribution of a fixed admissible partition $\pi \in \cP_\ell$ to the right-hand side of \eqref{eq:partition-exact}.
We first evaluate the hash factor exactly.

\begin{claim}[Exact hash expectation]\label{clm:hash-expectation}
For every admissible $\pi\in\cP_\ell$, every $\lambda\in\cL_\pi$, and every $a\in\cI_{\pi,\lambda}$,
\begin{equation}\label{eq:hash-expectation}
    \Hash_\pi(\lambda,a)=B^{\cc(\cC_\pi)-|\pi|}.
\end{equation}
\end{claim}

\begin{proof}
Consider one connected component of $\cC_\pi$ containing $b$ blocks. Along every edge of this component, the two endpoint blocks use the same layer and are required to hash to the same bucket. Connectivity therefore requires all $b$ hash values in the component to be equal. The $b$ coordinate labels are distinct because all blocks in the component use the same layer, so distinctness follows from \eqref{eq:row-injective}. Their hash values are independent and uniform in $[B]$, so the probability that they are all equal is $B^{1-b}$.

If several components use the same layer, their sets of coordinate labels are disjoint by \eqref{eq:row-injective}. Hence the hash variables used by different components are still independent. Components may land in the same bucket; no inequality between their bucket values is required. Multiplying $B^{1-b}$ over all $\cc(\cC_\pi)$ components, whose block counts sum to $|\pi|$, gives $B^{\cc(\cC_\pi)-|\pi|}$.
\end{proof}

Substituting \eqref{eq:hash-expectation} into \eqref{eq:partition-exact} gives
\begin{equation}\label{eq:partition-after-hash}
    \E\tr G^\ell
    =\sum_{\pi\in\cP_\ell}
      B^{\cc(\cC_\pi)-|\pi|}
      \sum_{\lambda\in\cL_\pi}
      \sum_{a\in\cI_{\pi,\lambda}}
      \Gram_\pi(a).
\end{equation}
The following lemma controls the innermost sum.

\begin{lemma}[Injective Eulerian Gram sum]\label{lem:injective-gram}
Let $u_1,\ldots,u_n\in\R^d$ satisfy $\sum_i u_i u_i^\top=\Id$. Let $H$ be an undirected multigraph with no isolated vertices and with every degree even. Let $\mathcal P$ be any partition of $V(H)$. Then
\begin{equation}\label{eq:injective-gram}
\left|
    \sum_{\substack{a:V(H)\to[n]\\
    a\text{ injective on every }P\in\mathcal P}}
    \prod_{\{v,w\}\in E(H)}
    \langle u_{a(v)},u_{a(w)}\rangle
\right|
\le (2e)^{|V(H)|}d^{\cc(H)}.
\end{equation}
\end{lemma}

We defer the proof of Lemma~\ref{lem:injective-gram} to Section~\ref{sec:gram}. For a fixed $\lambda$, group the vertices of $\cG_\pi$ by their actual layer. The condition $a\in\cI_{\pi,\lambda}$ says exactly that the coordinate labels are injective inside each layer group. By \eqref{eq:equal-degrees}, $\cG_\pi$ is Eulerian and has $\cc(\cG_\pi)$ connected components. Therefore
\begin{equation*}
    \left|
      \sum_{a\in\cI_{\pi,\lambda}}\Gram_\pi(a)
    \right|
    \le (2e)^{|\pi|} d^{\cc(\cG_\pi)}.
\end{equation*}
Define $\Cont(\pi)$ to be the contribution of the partition $\pi$ in \eqref{eq:partition-after-hash}. Using \eqref{eq:lambda-count},
\begin{equation}\label{eq:fixed-partition-raw}
    |\Cont(\pi)|
    \le s^{\cc(\cC_\pi)} B^{\cc(\cC_\pi)-|\pi|}(2e)^{|\pi|} d^{\cc(\cG_\pi)}.
\end{equation}

By \eqref{eq:technical-regime},
\[
    s\le\frac{p}{4\eps},
    \qquad
    B\ge\frac{8d}{\eps p}.
\]
Since $\cc(\cC_\pi)-|\pi|\le0$, $|\pi|-2\cc(\cC_\pi)\ge0$ by
\eqref{eq:component-block-count}, and $0<\eps\le1$, we obtain
\begin{equation}\label{eq:parameter-collapse}
\begin{aligned}
    &(2e)^{|\pi|}d^{\cc(\cG_\pi)}
      s^{\cc(\cC_\pi)}B^{\cc(\cC_\pi)-|\pi|} \\
    &\qquad\le
      \frac{(2e)^{|\pi|}}
      {4^{\cc(\cC_\pi)}8^{|\pi|-\cc(\cC_\pi)}}
      d^{\cc(\cG_\pi)+\cc(\cC_\pi)-|\pi|}
      p^{|\pi|}\eps^{\,|\pi|-2\cc(\cC_\pi)} \\
    &\qquad\le
      d^{\cc(\cG_\pi)+\cc(\cC_\pi)-|\pi|}p^{|\pi|}.
\end{aligned}
\end{equation}
We used
\[
    \frac{(2e)^{|\pi|}}
    {4^{\cc(\cC_\pi)}8^{|\pi|-\cc(\cC_\pi)}}
    =\left(\frac e4\right)^{|\pi|}2^{\cc(\cC_\pi)}
    \le\left(\frac{e}{2\sqrt2}\right)^{|\pi|}
    \le1.
\]

If $\ell$ is even, then $G$ is symmetric and $\tr G^\ell\ge0$ for every outcome. Hence
\begin{equation}\label{eq:triangle-after-exact}
    \E\tr G^\ell
    =\left|\sum_{\pi\in\cP_\ell}\Cont(\pi)\right|
    \le\sum_{\pi\in\cP_\ell}|\Cont(\pi)|.
\end{equation}

The power of $d$ in
\eqref{eq:parameter-collapse} is controlled by the combination
$\cc(\cG_\pi)+\cc(\cC_\pi)-|\pi|$. The following construction combines the components of
$\cC_\pi$ and $\cG_\pi$ into a single bipartite multigraph, where this
combination can be expressed using the numbers of vertices and edges.

\begin{samepage}
\begin{definition}[Component-incidence multigraph]\label{def:incidence}
The multigraph $\cT_\pi$ is the bipartite multigraph defined as follows.
\begin{itemize}
    \item Its left vertices are the connected components of $\cC_\pi$.
    \item Its right vertices are the connected components of $\cG_\pi$.
    \item Each block $D\in\pi$ gives one edge, its block-edge, joining the $\cC_\pi$-component and the $\cG_\pi$-component that contain $D$.
\end{itemize}
Parallel edges are allowed.
\end{definition}
\end{samepage}

Every block of $\pi$ becomes exactly one edge of $\cT_\pi$. Figure~\ref{fig:incidence} illustrates this contraction for the partition
\[
    D_1=\{L_1,L_2\},\quad
    D_2=\{R_1,R_3\},\quad
    D_3=\{R_2,L_3\},\quad
    D_4=\{L_4,R_5\},\quad
    D_5=\{R_4,L_5\}.
\]

\begin{figure}[t]
\centering
\begin{tikzpicture}[
  x=1cm,y=1cm,
  block/.style={draw,rounded corners=1.5pt,minimum width=1.55cm,minimum height=.62cm,align=center,font=\scriptsize,inner sep=1.5pt},
  component/.style={draw,dashed,rounded corners=5pt,line width=.45pt},
  internal/.style={line width=.45pt},
  contraction/.style={-{Latex[length=2mm]},line width=.75pt},
  tvertex/.style={draw,rounded corners=3pt,minimum width=2.75cm,minimum height=.95cm,align=center,inner sep=2pt,font=\scriptsize},
  tedge/.style={line width=.75pt},
  edge label/.style={font=\scriptsize,fill=white,inner sep=1pt},
  panel title/.style={font=\small,align=center}
]
\node[panel title] at (-4.8,4.55) {connected components of $\cC_\pi$};
\node[panel title] at (4.65,4.55) {connected components of $\cG_\pi$};

\node[component,minimum width=4.5cm,minimum height=2.2cm] (Ac1) at (-4.8,2.75) {};
\node[font=\small,anchor=south west,fill=white,inner sep=.5pt] at ([xshift=1mm,yshift=.4mm]Ac1.north west) {$\mathsf C_1$};
\node[block] (aD1) at (-5.95,2.80) {$D_1$\\[-1pt]$L_1,L_2$};
\node[block] (aD2) at (-4.00,3.18) {$D_2$\\[-1pt]$R_1,R_3$};
\node[block] (aD3) at (-4.00,2.25) {$D_3$\\[-1pt]$R_2,L_3$};
\draw[internal] (aD1)--(aD2);
\draw[internal] (aD1)--(aD3);
\draw[internal] (aD2)--(aD3);

\node[component,minimum width=3.9cm,minimum height=1.65cm] (Ac2) at (-4.8,.20) {};
\node[font=\small,anchor=south west,fill=white,inner sep=.5pt] at ([xshift=1mm,yshift=.4mm]Ac2.north west) {$\mathsf C_2$};
\node[block] (aD4) at (-5.70,.20) {$D_4$\\[-1pt]$L_4,R_5$};
\node[block] (aD5) at (-3.90,.20) {$D_5$\\[-1pt]$R_4,L_5$};
\draw[internal] (aD4) to[bend left=18] (aD5);
\draw[internal] (aD4) to[bend right=18] (aD5);

\node[component,minimum width=4.5cm,minimum height=2.2cm] (Bc1) at (4.65,2.75) {};
\node[font=\small,anchor=south west,fill=white,inner sep=.5pt] at ([xshift=1mm,yshift=.4mm]Bc1.north west) {$\mathsf G_1$};
\node[block] (bD1) at (3.55,3.15) {$D_1$\\[-1pt]$L_1,L_2$};
\node[block] (bD2) at (5.55,3.15) {$D_2$\\[-1pt]$R_1,R_3$};
\node[block] (bD4) at (4.55,2.20) {$D_4$\\[-1pt]$L_4,R_5$};
\draw[internal] (bD1)--(bD2)--(bD4)--(bD1);

\node[component,minimum width=2.05cm,minimum height=1.65cm] (Bc2) at (3.45,.20) {};
\node[font=\small,anchor=south west,fill=white,inner sep=.5pt] at ([xshift=1mm,yshift=.4mm]Bc2.north west) {$\mathsf G_2$};
\node[block] (bD3) at (3.45,.15) {$D_3$\\[-1pt]$R_2,L_3$};
\draw[internal] ([xshift=-1mm]bD3.north west)
  .. controls +(0,.32) and +(0,.32) .. ([xshift=1mm]bD3.north east);

\node[component,minimum width=2.05cm,minimum height=1.65cm] (Bc3) at (5.85,.20) {};
\node[font=\small,anchor=south west,fill=white,inner sep=.5pt] at ([xshift=1mm,yshift=.4mm]Bc3.north west) {$\mathsf G_3$};
\node[block] (bD5) at (5.85,.15) {$D_5$\\[-1pt]$R_4,L_5$};
\draw[internal] ([xshift=-1mm]bD5.north west)
  .. controls +(0,.32) and +(0,.32) .. ([xshift=1mm]bD5.north east);

\node[font=\small,align=center] at (0,-1.10) {contract each dashed connected component\\to one vertex};
\draw[contraction] (0,-1.55)--(0,-2.03);
\node[panel title] at (0,-2.32) {the component-incidence multigraph $\cT_\pi$};

\node[tvertex] (TA1) at (-3.55,-3.90) {$\mathsf C_1$\\$\{D_1,D_2,D_3\}$};
\node[tvertex] (TA2) at (-3.55,-5.70) {$\mathsf C_2$\\$\{D_4,D_5\}$};
\node[tvertex] (TB2) at (3.55,-3.10) {$\mathsf G_2$\\$\{D_3\}$};
\node[tvertex] (TB1) at (3.55,-4.60) {$\mathsf G_1$\\$\{D_1,D_2,D_4\}$};
\node[tvertex] (TB3) at (3.55,-6.10) {$\mathsf G_3$\\$\{D_5\}$};

\draw[tedge] (TA1.east) to[bend left=16] node[edge label,pos=.43,auto] {$D_1$} (TB1.west);
\draw[tedge] (TA1.east) to[bend right=9] node[edge label,pos=.52,auto] {$D_2$} (TB1.west);
\draw[tedge] (TA1.east)--node[edge label,pos=.70,auto] {$D_3$}(TB2.west);
\draw[tedge] (TA2.east)--node[edge label,pos=.51,auto] {$D_4$}(TB1.west);
\draw[tedge] (TA2.east)--node[edge label,pos=.55,auto] {$D_5$}(TB3.west);
\end{tikzpicture}
\caption{Construction of $\cT_\pi$ from an admissible partition. In the upper panel, each dashed region is a connected component of $\cC_\pi$ or $\cG_\pi$. Contracting those components produces the lower panel, where every block $D_j$ is exactly one edge of $\cT_\pi$.}
\label{fig:incidence}
\end{figure}
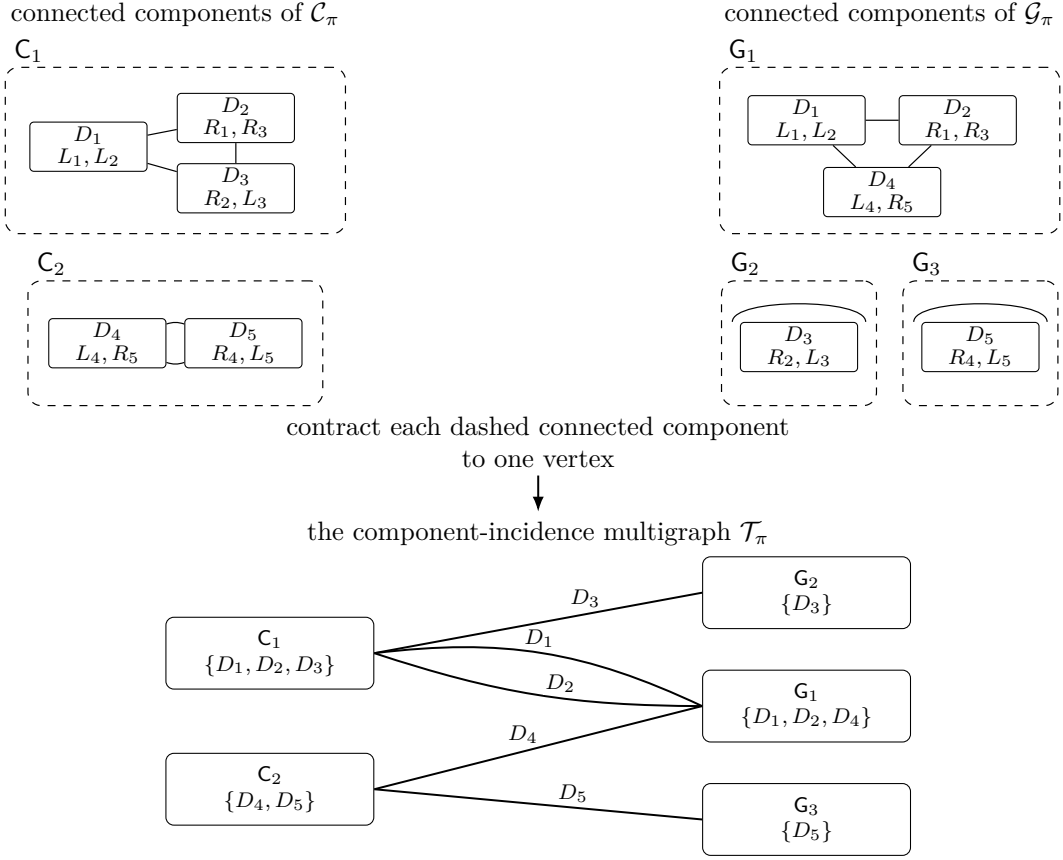

The multigraph $\cT_\pi$ is connected. Indeed, traverse in order
the block-edges corresponding to
\[
    L_1,R_1,L_2,R_2,\ldots,L_\ell,R_\ell.
\]
The edges corresponding to $L_t$ and $R_t$ meet at their common
$\cC_\pi$-component, while those corresponding to $R_t$ and $L_{t+1}$ meet
at their common $\cG_\pi$-component, where $L_{\ell+1}=L_1$. This gives a
closed walk containing every block-edge of $\cT_\pi$, because every block
contains at least one trace position. Thus $\cT_\pi$ is connected.

The multigraph has $|\pi|$ edges and $\cc(\cC_\pi)+\cc(\cG_\pi)$ vertices. Define the cycle rank
\begin{equation*}
    g_\pi:=g(\cT_\pi)
    :=|E(\cT_\pi)|-|V(\cT_\pi)|+1
    =|\pi|-(\cc(\cC_\pi)+\cc(\cG_\pi))+1.
\end{equation*}
For a connected graph, the cycle rank is the dimension of its cycle space, equivalently the number of independent cycles. In particular,
\begin{equation*}
    \cc(\cC_\pi)+\cc(\cG_\pi)-|\pi|=1-g_\pi.
\end{equation*}
Combining with \eqref{eq:fixed-partition-raw} and \eqref{eq:parameter-collapse},
\begin{equation}\label{eq:fixed-partition-final}
    |\Cont(\pi)|
    \le d^{1-g_\pi}p^{|\pi|}.
\end{equation}
The power of $d$ has two sources: the Gram sum contributes $d^{\cc(\cG_\pi)}$, while the hash probability contributes $d^{\cc(\cC_\pi)-|\pi|}$. The component-incidence multigraph combines them into $d^{1-g_\pi}$, so every independent cycle contributes one additional factor $1/d$.

\subsection{Canonical walks}

We encode each component-incidence multigraph $\cT_\pi$ by a
canonical walk. This bounds the sum over admissible partitions by a weighted
walk count whose supports split into attached trees and a $2$-core.

\begin{definition}[Canonical walk]\label{def:canonical-walk}
A canonical walk of length $2\ell$ is a closed sequence
\[
    W=(v_0,e_1,v_1,e_2,v_2,\ldots,e_{2\ell},v_{2\ell})
\]
with the following properties.
\begin{enumerate}[label=(\roman*)]
    \item $v_{2\ell}=v_0$.
    \item Each vertex name and edge name is a positive integer.
    \item Every edge name has one fixed unordered pair of endpoint names. Whenever $e_t$ appears, its endpoints are $v_{t-1}$ and $v_t$.
\end{enumerate}
Different edges may have the same endpoints, so parallel edges are allowed. The walk is canonical when $v_0=1$, new vertex names appear as $2,3,\ldots$ in order of first appearance, and new edge names appear as $1,2,\ldots$ in order of first appearance.

Its support contains the vertices and edges that occur in the sequence. Define its number of support edges and its cycle rank by
\begin{align*}
    e(W)&=\#\,\{\text{edges in support of } W\},  \\
    g(W)&=\#\,\{\text{edges in support of } W\} -
    \#\,\{\text{vertices in support of } W\}+1.
\end{align*}
Let $\cW_\ell$ be the set of canonical walks of length $2\ell$ whose support contains no self-loops and in which every edge name occurs a positive even number of times.
\end{definition}

An admissible partition $\pi\in\cP_\ell$ gives such a walk.
Read the trace positions cyclically in the order
\[
    L_1,R_1,L_2,R_2,\ldots,L_\ell,R_\ell,
\]
and replace each position by the block-edge of $\cT_\pi$ corresponding to
its block. Consecutive block-edges meet: those corresponding to $L_t$ and
$R_t$ share their $\cC_\pi$-component, while those corresponding to $R_t$
and $L_{t+1}$ share their $\cG_\pi$-component, with indices interpreted
cyclically. Thus this edge sequence forms a closed walk of length $2\ell$
on $\cT_\pi$.

The block-edge corresponding to a block $D$ is traversed exactly $|D|$
times, which is a positive even number because $\pi$ is admissible. Since
$\cT_\pi$ is bipartite, its support has no self-loops. Relabeling vertices
and edges in order of first appearance therefore produces a canonical walk
$W\in\cW_\ell$.

The map $\pi\mapsto W$ is injective: two trace positions belong to the same
block of $\pi$ exactly when the corresponding traversals have the same edge
name, so the edge-name equality classes recover $\pi$. Every block-edge
occurs, and hence the support of $W$ is $\cT_\pi$. Therefore
\[
    e(W)=|\pi|,
    \qquad
    g(W)=g_\pi.
\]

To obtain an upper bound, we may discard the bipartite structure and the alternation constraint, since doing so only enlarges the class of walks. Equations \eqref{eq:triangle-after-exact} and \eqref{eq:fixed-partition-final} therefore yield the following proposition.

\begin{proposition}[Trace tuples reduce to weighted canonical walks]\label{prop:trace-to-walks}
For every even integer $\ell\ge2$,
\begin{equation*}
    \E\tr G^\ell
    \le d\sum_{W\in\cW_\ell}d^{-g(W)}p^{e(W)}.
\end{equation*}
\end{proposition}

\begin{proof}
By \eqref{eq:triangle-after-exact} and \eqref{eq:fixed-partition-final},
\[
    \E\tr G^\ell
    \le\sum_{\pi\in\cP_\ell}|\Cont(\pi)|
    \le\sum_{\pi\in\cP_\ell}d^{1-g_\pi}p^{|\pi|}.
\]
The map $\pi\mapsto W$ constructed above sends each $\pi\in\cP_\ell$ to a canonical walk $W\in\cW_\ell$ with $g(W)=g_\pi$ and $e(W)=|\pi|$, and is injective. Summing the same summand $d^{1-g(W)}p^{e(W)}$ over the larger set $\cW_\ell$ in place of the image of $\cP_\ell$ can only increase the total, so
\[
    \sum_{\pi\in\cP_\ell}d^{1-g_\pi}p^{|\pi|}
    \le\sum_{W\in\cW_\ell}d^{1-g(W)}p^{e(W)}
    =d\sum_{W\in\cW_\ell}d^{-g(W)}p^{e(W)}.
\]
\end{proof}

The following combinatorial lemma bounds this weighted sum on canonical walks. The proof is given in Section~\ref{sec:walk-count}.

\begin{lemma}[Weighted canonical-walk bound]\label{lem:walk-count}
Let $C_3=18$ and $C_4=\log 16$. For $d\ge2$, even integer $\ell$ satisfying $2 \le \ell \le \log d/C_4$, and $x\ge\ell$,
\begin{equation*}
    \sum_{W\in\cW_\ell}d^{-g(W)}x^{e(W)}
    \le (C_3x)^\ell.
\end{equation*}
\end{lemma}

This gives the required moment bound.

\begin{proposition}[Logarithmic trace moment]\label{prop:trace-moment}
For the constants $C_3,C_4$ in Lemma~\ref{lem:walk-count}, every $d\ge2$ and every even integer $\ell\ge2$ satisfying $\ell\le \log d/C_4$, we have
\begin{equation*}
    \E\tr G^\ell\le d(C_3p)^\ell.
\end{equation*}
In addition, for every $d\ge2$,
\begin{equation*}
    \E\tr G^2\le2dp^2.
\end{equation*}
\end{proposition}

\begin{proof}
Proposition~\ref{prop:trace-to-walks} gives
\[
    \E\tr G^\ell
    \le d\sum_{W\in\cW_\ell}d^{-g(W)}p^{e(W)}.
\]
Since $p\ge\log d\ge\ell$, Lemma~\ref{lem:walk-count} applies with $x=p$ and gives
\[
    \E\tr G^\ell\le d(C_3p)^\ell.
\]
The second-moment bound is proved in Appendix~\ref{app:small-d}.
\end{proof}

\subsection{Completion of the mean bound}

\begin{proof}[Proof of Proposition~\ref{prop:mean-error}]
Assume first that $d\ge16^4$. Let $L=\log_{16}d$, and let $\ell_0$ be the largest even integer at most $L$. Then $\ell_0\ge4$ and
\[
    L<\ell_0+2\le\frac32\ell_0,
    \qquad
    d^{1/\ell_0}=16^{L/\ell_0}<64.
\]
Proposition~\ref{prop:trace-moment} and $p/s\le8\eps$ therefore give
\[
\begin{aligned}
    \E\lVert E\rVert_{\op}
    &\le \bigl(\E\lVert E\rVert_{\op}^{\ell_0}\bigr)^{1/\ell_0}
     \le \bigl(\E\tr E^{\ell_0}\bigr)^{1/\ell_0}\\
    &=\frac1s\bigl(\E\tr G^{\ell_0}\bigr)^{1/\ell_0}
     <64C_3\frac ps
     \le512C_3\eps.
\end{aligned}
\]

If $2\le d<16^4$, the second-moment bound in Proposition~\ref{prop:trace-moment} gives
\[
    \E\lVert E\rVert_{\op}
    \le\frac1s\bigl(\E\tr G^2\bigr)^{1/2}
    \le\sqrt{2d}\,\frac ps
    <2048\sqrt2\,\eps.
\]
Since $C_3=18$, choosing $C_1=9216$ proves the proposition.
\end{proof}

The next two sections prove the lemmas used above.

\section{Injective Eulerian Gram sums}\label{sec:gram}

The goal of this section is to prove Lemma~\ref{lem:injective-gram}. One challenge is that partitions impose injectivity constraints on certain groups of vertices. We first prove a restricted Eulerian contraction bound with arbitrary vertex-dependent label sets, and then use inclusion--exclusion to enforce these injectivity constraints and obtain Lemma~\ref{lem:injective-gram}.

Throughout this section, $u_1,\ldots,u_n\in\R^d$ form a Parseval frame:
\begin{equation*}
    \sum_{i=1}^n u_i u_i^\top=\Id.
\end{equation*}
Loops in an undirected multigraph count twice toward degree.

\subsection{Restricted Eulerian contraction}

An oriented generalized edge $e=(v,w)$ carries a matrix $M_e$ with $\lVert M_e\rVert_{\op}\le1$ and contributes $\langle u_{a(v)},M_eu_{a(w)}\rangle$. Reversing the orientation replaces $M_e$ by $M_e^\top$. Ordinary Gram factors correspond to $M_e=\Id$.

\begin{lemma}[Restricted Eulerian contraction]\label{lem:restricted-contraction}
Let $H$ be an undirected multigraph with no isolated vertices and with every degree even. For each vertex $v$, choose an arbitrary set $A_v\subseteq[n]$. Then
\begin{equation*}
\left|
    \sum_{a(v)\in A_v \\ (v\in V(H))}
    \prod_{e=(v,w)\in E(H)}
    \langle u_{a(v)},M_eu_{a(w)}\rangle
\right|
\le d^{\cc(H)}.
\end{equation*}
\end{lemma}

The statement includes ordinary Gram products by taking every $M_e=\Id$.

\paragraph{Proof sketch.}
It is enough to treat one connected component. If the graph has two edge-disjoint spanning trees $\cT_1,\cT_2$, all other edge factors have absolute value at most one. For a fixed labeling, let $A$ and $B$ be the products along $\cT_1$ and $\cT_2$. Then
\[
    2|AB|\le |A|^2+|B|^2.
\]
Each squared tree sum is at most $d$: remove leaves one at a time and use
\[
    \sum_{i\in A_v}|\langle u_i,Mx\rangle|^2
    =x^\top M^\top\!\left(\sum_{i\in A_v}u_i u_i^\top\right)Mx
    \le\lVert x\rVert_2^2.
\]
If there are no two edge-disjoint spanning trees, Eulerian parity and the Nash--Williams criterion give a two-edge cut. Choose a smallest shore of such a cut. That shore contains two edge-disjoint spanning trees, so summing all labels on the shore produces a bilinear form $\langle x,My\rangle$ with $\lVert M\rVert_{\op}\le1$. Replace the whole shore by this generalized edge and continue by induction. The components factor, giving $d^{\cc(H)}$. Appendix~\ref{app:contraction} contains the full proof.

\subsection{Inclusion--exclusion and proof of Lemma~\ref{lem:injective-gram}}

\begin{proof}[Proof of Lemma~\ref{lem:injective-gram}]
Fix one block $P\in\mathcal P$ of size $t$. Let $\theta:[n]\to[t]$ be a uniformly random map. For a fixed labeling $a$ on $P$, inclusion--exclusion gives
\begin{equation*}
    \one\bigl\{\{\theta(a(v)):v\in P\}=[t]\bigr\}
    =\sum_{F\subseteq[t]}(-1)^{t-|F|}
      \prod_{v\in P}\one\{\theta(a(v))\in F\}.
\end{equation*}
If the labels $a(v)$ are distinct, the values $\theta(a(v))$ are independent and uniform, and they cover $[t]$ with probability $t!/t^t$. If a label repeats, they cannot cover all $t$ values. Hence
\begin{align}
    \one\bigl\{\{a(v):v\in P\}\text{ are distinct}\bigr\}
    &=\frac{t^t}{t!}
      \E_\theta\one\bigl\{\{\theta(a(v)):v\in P\}=[t]\bigr\}\notag\\
    &=\frac{t^t}{t!}
      \E_\theta\sum_{F\subseteq[t]}(-1)^{t-|F|}
      \prod_{v\in P}\one\{\theta(a(v))\in F\}.
      \label{eq:IE-injectivity}
\end{align}
Apply \eqref{eq:IE-injectivity} independently to every block. After fixing all maps $\theta$ and all subsets $F$, each vertex $v\in P$ is simply restricted to an allowed set
\[
    A_v=\theta^{-1}(F).
\]
Thus each term is a restricted Eulerian Gram sum, so Lemma~\ref{lem:restricted-contraction} bounds its absolute value by $d^{\cc(H)}$.

A block of size $t$ contributes $2^t$ inclusion--exclusion terms and the factor $t^t/t!$. Since $t!\ge(t/e)^t$,
\[
    \prod_{P\in\mathcal P}
    2^{|P|}\frac{|P|^{|P|}}{|P|!}
    \le
    \prod_{P\in\mathcal P}(2e)^{|P|}
    =(2e)^{|V(H)|}.
\]
Taking absolute values and expectations proves \eqref{eq:injective-gram}.
\end{proof}

\section{A weighted count of canonical walks}\label{sec:walk-count}

The goal of this section is to prove Lemma~\ref{lem:walk-count}, which requires only Definition~\ref{def:canonical-walk}. We show that the factor $d^{-g(W)}$ compensates for the additional combinatorial complexity created by cycles in the support of a canonical walk. We first count walks whose support is a tree, where the traversal is encoded by a Dyck word together with the choices that create or revisit edges. For positive cycle rank, we prune all attached trees and encode the remaining $2$-core route. If every core vertex has degree two, the core is a cycle. Otherwise, the cycle rank controls both the maximum core degree and the number of first traversals that close a cycle. Since the moment $\ell$ is $O(\log d)$, the cycle-rank penalty can absorb the resulting core-route count.

We assign to each walk a codeword that records enough information to
reconstruct it. The map from walks to codewords is injective, and the
local code symbols are weighted so that the codeword associated with
$W$ has weight $x^{e(W)}$. Therefore, it suffices to bound the total
weight of all possible codewords.

During encoding and decoding, whenever a vertex or edge is first
discovered, it receives the next available temporary label according to a fixed
convention shared by the encoder and decoder. These temporary labels are
distinct from the canonical vertex and edge names. After reconstructing the
labeled walk, the decoder applies the canonical relabeling from
Definition~\ref{def:canonical-walk}.

Let $\cW_{\ell,g}$ be the walks in $\cW_\ell$ whose support has cycle rank $g$, and let
\begin{equation*}
    A_{\ell,g}(x)=\sum_{W\in\cW_{\ell,g}}x^{e(W)}.
\end{equation*}
Every support edge is used at least twice, so
\begin{equation*}
    e(W)\le \ell,
    \qquad
    0\le g(W)\le e(W)\le \ell.
\end{equation*}
We first count tree supports and then supports with positive cycle rank.

\subsection{Tree supports}

Assume $g=0$. The support is a tree rooted at $v_0$. Record an up-step when the distance from the root increases and a down-step when it decreases. The result is a Dyck word of length $2\ell$, so there are at most $4^\ell$ possibilities.

For a fixed Dyck word, every down-step is forced: it uses the parent edge. At an up-step, the walk either uses a previously discovered child edge or creates a new child edge. There are at most $\ell$ old edges, and there is one new-edge symbol. Give an old-edge choice weight $1$ and a new-edge choice weight $x$. Each support edge is created once, so the product of local weights is $x^{e(W)}$. Therefore, when $x\ge \ell$,
\begin{equation}\label{eq:tree-count}
    A_{\ell,0}(x)
    \le
    \underbrace{\mathstrut 4^\ell}_{\text{Dyck words}}
    \;\cdot\;
    \underbrace{\mathstrut (\ell+x)^\ell}_{\text{up-step records}}
    \le (8x)^\ell.
\end{equation}
This is an injective weighted code. The Dyck word gives the up- and down-steps, an old choice records the name of a previously discovered child edge, a new choice creates a fresh edge and vertex, and a down-step uses the forced parent edge.

\subsection{Pruning attached trees}

Fix $g\ge1$ and let $H$ be the support. Repeatedly delete degree-one vertices. The remaining graph $H^{(2)}$ is a nonempty connected $2$-core. Deleting a leaf removes one vertex and one edge, so it preserves $g$. The deleted edges form rooted trees, each attached to exactly one core vertex.

\begin{figure}[H]
\centering
\begin{tikzpicture}[
  x=.86cm,y=.86cm,
  vertex/.style={
    circle,
    draw,
    fill=white,
    inner sep=0pt,
    minimum size=4.2mm
  },
  core edge/.style={line width=.85pt},
  forest edge/.style={line width=.42pt},
  arrow/.style={
    -{Latex[length=2mm]},
    line width=.75pt
  },
  panel label/.style={
    font=\scriptsize,
    align=center
  },
  operation/.style={
    font=\scriptsize,
    align=center,
    fill=white,
    inner sep=1.2pt
  }
]

\node[vertex] (u)  at (1.70,0) {};
\node[vertex] (v)  at (5.30,0) {};
\node[vertex] (t)  at (3.50,1.20) {};
\node[vertex] (m1) at (2.90,0) {};
\node[vertex] (m2) at (4.10,0) {};
\node[vertex] (b1) at (2.90,-1.20) {};
\node[vertex] (b2) at (4.10,-1.20) {};

\draw[core edge] (u)--(t)--(v);
\draw[core edge] (u)--(m1)--(m2)--(v);
\draw[core edge] (u)--(b1)--(b2)--(v);

\node[vertex] (l1) at (.85,.60) {};
\node[vertex] (l2) at (.15,1.15) {};
\node[vertex] (l3) at (.15,.05) {};

\draw[forest edge] (u)--(l1);
\draw[forest edge] (l1)--(l2);
\draw[forest edge] (l1)--(l3);

\node[vertex] (t1) at (3.50,2.05) {};
\node[vertex] (t2) at (4.15,2.65) {};

\draw[forest edge] (t)--(t1)--(t2);

\node[vertex] (r1) at (6.05,.65) {};
\node[vertex] (r2) at (6.05,1.50) {};

\draw[forest edge] (v)--(r1)--(r2);

\node[panel label] at (3.10,-2.15)
  {support $H$\\attached trees shown by thin edges};

\draw[arrow] (6.70,0)--(8.00,0);

\node[operation] at (7.35, .50)
  {prune leaves};

\node[vertex] (u2)  at (8.65,0) {};
\node[vertex] (v2)  at (12.25,0) {};
\node[vertex] (t2c) at (10.45,1.20) {};
\node[vertex] (m12) at (9.85,0) {};
\node[vertex] (m22) at (11.05,0) {};
\node[vertex] (b12) at (9.85,-1.20) {};
\node[vertex] (b22) at (11.05,-1.20) {};

\draw[core edge] (u2)--(t2c)--(v2);
\draw[core edge] (u2)--(m12)--(m22)--(v2);
\draw[core edge] (u2)--(b12)--(b22)--(v2);

\node[panel label] at (10.45,-2.15)
  {$2$-core $H^{(2)}$\\attached trees removed};

\end{tikzpicture}
\caption{Pruning attached trees. Repeatedly deleting degree-one
vertices removes the attached trees and preserves the cycle rank,
leaving the nonempty connected $2$-core $H^{(2)}$.}
\label{fig:core-pruning}
\end{figure}
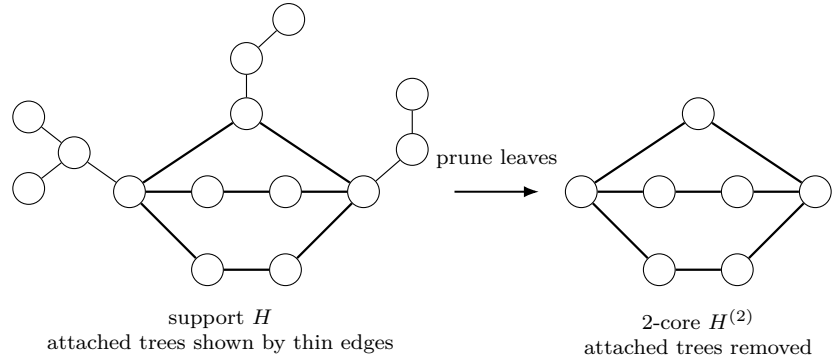

Rotate the closed walk so that it starts at its first visit to the core, and record the rotation offset. The rotated walk together with the offset uniquely determines the original walk. Erase every maximal excursion into an attached tree. The remaining route is a closed walk on $H^{(2)}$ of length $2\ell_c$. Let $\ell_f=\ell-\ell_c$. The erased forest part has $\ell_f$ outward steps and $\ell_f$ inward steps.

We use a mask to record whether each of the $2\ell$ positions is a core step, an outward forest step, or an inward forest step. The mask has at most $3^{2\ell}=9^\ell$ possibilities. We include the factor $2\ell$ for the rotation offset in the common bound below.

\begin{claim}[Forest-extension code]\label{clm:forest-extension-code}
Fix a rotated core route of length $2\ell_c$, and let $\ell_f=\ell-\ell_c$. The mask, together with one symbol at every outward forest step, determines at most one rotated walk. The symbol is either
\[
    \mathsf{new},
    \qquad\text{or}\qquad
    \mathsf{old}(e),
\]
where $e$ names a previously created child edge incident to the current vertex. If $\mathsf{new}$ has weight $x$ and every $\mathsf{old}(e)$ has weight $1$, then the total weighted number of masks and forest-extension symbols is at most
\[
    \underbrace{\mathstrut 9^\ell}_{\text{masks}}
    \;\cdot\;
    \underbrace{\mathstrut (\ell+x)^{\ell_f}}_{\text{forest records}}.
\]
\end{claim}

\begin{proof}
Scan the mask from left to right while following the fixed core route. At a core position, take the next step of the core route. At an outward forest position, a symbol $\mathsf{new}$ creates a fresh edge and vertex, while $\mathsf{old}(e)$ traverses the named child edge. At an inward forest position, the parent edge is forced because every component outside the core is a rooted tree. Thus the data reconstruct at most one rotated walk.

There are at most $3^{2\ell}=9^\ell$ masks. At each of the $\ell_f$ outward positions there are at most $\ell$ old-edge choices of weight $1$ and one new-edge choice of weight $x$. The new-edge symbols are in bijection with the support edges outside the core, so their product weight is exactly the forest contribution to $x^{e(W)}$. This proves the claim.
\end{proof}

Claim~\ref{clm:forest-extension-code} gives
\begin{equation}\label{eq:forest-extension}
    9^\ell(\ell+x)^{\ell_f}
    \le 9^\ell(2x)^{\ell_f},
\end{equation}
since $x\ge\ell$.
If $e_c=|E(H^{(2)})|$, every core edge is used at least twice, so
\begin{equation}\label{eq:core-edge-weight}
    e_c\le \ell_c,
    \qquad
    x^{e_c}\le x^{\ell_c}.
\end{equation}
Combining \eqref{eq:forest-extension} and
\eqref{eq:core-edge-weight}, and using
$\ell_f+\ell_c=\ell$ and $\ell_f\le\ell$, we obtain, for each fixed core route, that the combined contribution of the mask choices,
the weighted forest-extension records, and the core-edge weight is at most
\[
    \underbrace{\mathstrut 9^\ell}_{\text{masks}}
    \;\cdot\;
    \underbrace{\mathstrut (2x)^{\ell_f}}_{\text{forest records}}
    \;\cdot\;
    \underbrace{\mathstrut x^{e_c}}_{\text{core weight}}
    \le 9^\ell 2^{\ell_f}x^{\ell_f+\ell_c}
    \le 9^\ell 2^\ell x^\ell
    =(18x)^\ell.
\]
Including the factor $2\ell$ for the rotation offset gives the common bound
\begin{equation}\label{eq:common-forest-core}
    2\ell
    \cdot
    (18x)^\ell
\end{equation}
In the next sections, we will multiply this common bound by the number of possible encodings of the core route. The following elementary facts control the core encoding.

\begin{claim}[Core structure]\label{clm:core-structure}
Let $v_c=|V(H^{(2)})|$ and $e_c=|E(H^{(2)})|$. Then $v_c\le e_c\le\ell$, and every core vertex has degree at most $2g$. If the starting vertex of the rotated core route is declared discovered, then, among the first traversals of the distinct core edges, exactly $g$ arrive at a previously discovered vertex.
\end{claim}

\begin{proof}
Since the core is connected and has cycle rank $g$,
\[
    \sum_{v\in V(H^{(2)})}(\deg(v)-2)
    =2e_c-2v_c
    =2g-2.
\]
Every core degree is at least two, so each summand is nonnegative and $\deg(v)\le2g$. Also $e_c=v_c-1+g\ge v_c$, while $e_c\le\ell_c\le\ell$ by \eqref{eq:core-edge-weight}.

Consider each core edge at the first time it is traversed. Exactly $v_c-1$ of these traversals reach a new vertex: every core vertex other than the starting vertex is first reached once, and the corresponding edges form a spanning tree. The remaining
\[
    e_c-(v_c-1)=g
\]
first traversals reach an already discovered vertex.
\end{proof}

\subsection{Encoding the degree-two core}

If every core vertex has degree two, then the core is a cycle, including the two-edge cycle formed by two parallel edges, and $g=1$. Store the cycle length and, at every core step, the name of the incident core edge traversed. At the current vertex there are exactly two incident core edges, counting the two parallel edges separately in the two-edge cycle. Hence the number of possible core encodings is at most
\[
    \underbrace{\mathstrut \ell}_{\text{cycle length}}
    \;\cdot\;
    \underbrace{\mathstrut 2^{2\ell_c}}_{\text{core-step records}}
    \le \ell\,4^\ell.
\]
Multiplying by the common bound in \eqref{eq:common-forest-core} gives
\begin{equation}\label{eq:unicyclic}
    A_{\ell,1}(x)
    \le
    \underbrace{\mathstrut 2\ell(18x)^\ell}_{\text{common bound}}
    \;\cdot\;
    \underbrace{\mathstrut \ell\,4^\ell}_{\text{core encodings}}
    =(18x)^\ell\,2\ell^2\,4^\ell.
\end{equation}

This core code is injective. Given the stored cycle length, construct a cycle with a distinguished starting vertex and distinct edge names; when the length is two, retain two distinct parallel edge names. Starting at the distinguished vertex, follow the named incident edge at each core step. This reconstructs the rotated core route. Claim~\ref{clm:forest-extension-code} reconstructs the attached-tree excursions. Finally, undo the recorded rotation and take the canonical representative of the resulting walk.

\subsection{Encoding a branching core}

Suppose that not every core vertex has degree two. Then $g\ge2$. Read the rotated core route from its starting vertex. At each core step, record one of the following:
\begin{enumerate}[label=(\roman*)]
    \item $\mathsf{old}$: traverse a previously exposed core edge incident to the current vertex, recording the name of that edge;
    \item $\mathsf{new}$: traverse a core edge for the first time toward a previously undiscovered vertex;
    \item $\mathsf{close}$: traverse a core edge for the first time toward a previously discovered vertex, recording the name of that endpoint separately.
\end{enumerate}
By Claim~\ref{clm:core-structure}, there are at most $2g$ choices in case~(i), and cases~(ii) and~(iii) each contribute one additional symbol ($\mathsf{new}$ or $\mathsf{close}$). Thus there are at most
\[
    \underbrace{\mathstrut 2g}_{\mathsf{old}}
    \quad+\quad
    \underbrace{\mathstrut 1}_{\mathsf{new}}
    \quad+\quad
    \underbrace{\mathstrut 1}_{\mathsf{close}}
    =2(g+1)
\]
possible step records before the endpoint data in case~(iii) are specified. Moreover, case~(iii) occurs exactly $g$ times, and each recorded endpoint has at most $v_c\le\ell$ choices. The core part of the code therefore has at most
\[
    \underbrace{\mathstrut \ell^g}_{\text{recorded endpoints}}
    \;\cdot\;
    \underbrace{\mathstrut [2(g+1)]^{2\ell_c}}_{\text{core-step records}}
    \le \ell^g[2(g+1)]^{2\ell}
\]
possibilities. Multiplying by the common bound in \eqref{eq:common-forest-core} gives
\begin{equation}\label{eq:Apg-core}
    A_{\ell,g}(x)
    \le
    \underbrace{\mathstrut 2\ell(18x)^\ell}_{\text{common bound}}
    \;\cdot\;
    \underbrace{\mathstrut \ell^g[2(g+1)]^{2\ell}}_{\text{core encodings}}
    =(18x)^\ell\,2\ell^{g+1}[2(g+1)]^{2\ell}
    \qquad (g\ge2).
\end{equation}

The core code is injective. Begin with the starting vertex and process the step records in order. In case~(i), follow the named exposed edge; in case~(ii), add a fresh edge and a fresh vertex; and in case~(iii), add a fresh edge to the named exposed vertex. The labeling convention makes every selected object unambiguous, and parallel edges remain distinct because they have distinct names. This reconstructs the rotated core route. Claim~\ref{clm:forest-extension-code} then reconstructs the attached-tree excursions. Finally, undo the recorded rotation and take the canonical representative of the resulting walk.

An example of the encoding and decoding process is given in Figure~\ref{fig:encoding-example}. In the example, we show how to build a codeword from a canonical walk, and then decode it to recover the original canonical walk. Note that not all codewords can be decoded to a canonical walk.

\begin{figure}[H]
\centering
\setlength{\fboxsep}{5pt}
\fbox{\begin{minipage}{0.92\linewidth}
\footnotesize
\setlength{\abovedisplayskip}{3pt}
\setlength{\belowdisplayskip}{3pt}
\setlength{\jot}{2pt}
\setlength{\tabcolsep}{3pt}
\renewcommand{\arraystretch}{1.08}

\textbf{1. Canonical walk.}\par
\smallskip

\noindent
\begin{minipage}[t]{0.63\linewidth}
\vspace{0pt}
Consider
\[
\begin{aligned}
W={}& v_1 \xrightarrow{e_1} v_2 \xrightarrow{e_2} v_3
      \xrightarrow{e_3} v_4 \xrightarrow{e_4} v_5
      \xrightarrow{e_4} v_4 \xrightarrow{e_3} v_3 \\
   & \xrightarrow{e_5} v_2 \xrightarrow{e_6} v_3
     \xrightarrow{e_2} v_2 \xrightarrow{e_5} v_3
     \xrightarrow{e_6} v_2 \xrightarrow{e_1} v_1 .
\end{aligned}
\]
The vertex and edge names appear in first-encounter order, so $W$ is
canonical. Here
\[
    \ell=6,
    \qquad
    e(W)=6,
    \qquad
    g(W)=6-5+1=2.
\]
\end{minipage}\hfill
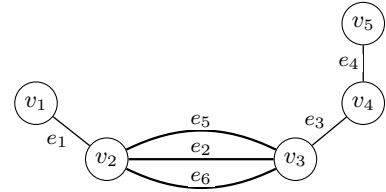
\begin{minipage}[t]{0.33\linewidth}
\vspace{0pt}
\centering
\begin{tikzpicture}[
  x=.78cm,y=.78cm,
  vertex/.style={
    circle,
    draw,
    fill=white,
    inner sep=0pt,
    minimum size=16pt
  },
  core edge/.style={line width=.85pt},
  forest edge/.style={line width=.5pt},
  elabel/.style={fill=white,inner sep=1pt,font=\scriptsize}
]
\node[vertex] (v2) at (0,0) {$v_2$};
\node[vertex] (v3) at (3.2,0) {$v_3$};
\node[vertex] (v1) at (-1.2,.95) {$v_1$};
\node[vertex] (v4) at (4.35,.95) {$v_4$};
\node[vertex] (v5) at (4.35,2.30) {$v_5$};

\draw[core edge] (v2) -- node[elabel,auto] {$e_2$} (v3);
\draw[core edge,bend left=26] (v2) to node[elabel,auto] {$e_5$} (v3);
\draw[core edge,bend right=26] (v2) to node[elabel,auto] {$e_6$} (v3);

\draw[forest edge] (v2) -- node[elabel,auto] {$e_1$} (v1);
\draw[forest edge] (v3) -- node[elabel,auto] {$e_3$} (v4);
\draw[forest edge] (v4) -- node[elabel,auto] {$e_4$} (v5);
\end{tikzpicture}
\end{minipage}

\medskip
\textbf{2. Codeword.}
Rotate by one edge-step so that the walk begins at its first core
vertex $v_2$. Along the rotated walk, the encoder labels vertices and
edges in order of first appearance. This yields
\[
(v_2,v_3,v_4,v_5,v_1)
\mapsto
(\widehat v_1,\widehat v_2,\widehat v_3,\widehat v_4,\widehat v_5),
\]
\[
(e_2,e_3,e_4,e_5,e_6,e_1)
\mapsto
(\widehat e_1,\widehat e_2,\widehat e_3,\widehat e_4,\widehat e_5,\widehat e_6).
\]
The codeword is
\[
\mathcal C(W)=
\left[
\begin{array}{rcl}
r
&=&1,\\
M
&=&(\mathsf C,\mathsf O,\mathsf O,\mathsf I,\mathsf I,
     \mathsf C,\mathsf C,\mathsf C,\mathsf C,\mathsf C,
     \mathsf O,\mathsf I),\\
R_c
&=&(\mathsf{new},
     \mathsf{close}(\widehat v_1),
     \mathsf{close}(\widehat v_2),
     \mathsf{old}(\widehat e_1),
     \mathsf{old}(\widehat e_4),
     \mathsf{old}(\widehat e_5)),\\
R_f
&=&(\mathsf{new},\mathsf{new},\mathsf{new}).
\end{array}
\right]
\]
Here $r$ is the rotation offset, $M$ is the mask, and $R_c,R_f$ are the
core and outward-forest records.

\medskip
\textbf{3. Decoding.}
Start at $\widehat v_1$ and scan $M$ from left to right, consuming the
next record from $R_c$ or $R_f$ as appropriate.

\begin{center}
\scriptsize
\begin{tabular}{ccl@{\qquad}ccl}
$t$ & record & reconstructed step & $t$ & record & reconstructed step \\ \hline
1  & $\mathsf C\!:\!\mathsf{new}$
   & $\widehat v_1\xrightarrow{\widehat e_1}\widehat v_2$
   & 7  & $\mathsf C\!:\!\mathsf{close}(\widehat v_2)$
       & $\widehat v_1\xrightarrow{\widehat e_5}\widehat v_2$\\
2  & $\mathsf O\!:\!\mathsf{new}$
   & $\widehat v_2\xrightarrow{\widehat e_2}\widehat v_3$
   & 8  & $\mathsf C\!:\!\mathsf{old}(\widehat e_1)$
       & $\widehat v_2\xrightarrow{\widehat e_1}\widehat v_1$\\
3  & $\mathsf O\!:\!\mathsf{new}$
   & $\widehat v_3\xrightarrow{\widehat e_3}\widehat v_4$
   & 9  & $\mathsf C\!:\!\mathsf{old}(\widehat e_4)$
       & $\widehat v_1\xrightarrow{\widehat e_4}\widehat v_2$\\
4  & $\mathsf I$
   & $\widehat v_4\xrightarrow{\widehat e_3}\widehat v_3$
   & 10 & $\mathsf C\!:\!\mathsf{old}(\widehat e_5)$
       & $\widehat v_2\xrightarrow{\widehat e_5}\widehat v_1$\\
5  & $\mathsf I$
   & $\widehat v_3\xrightarrow{\widehat e_2}\widehat v_2$
   & 11 & $\mathsf O\!:\!\mathsf{new}$
       & $\widehat v_1\xrightarrow{\widehat e_6}\widehat v_5$\\
6  & $\mathsf C\!:\!\mathsf{close}(\widehat v_1)$
   & $\widehat v_2\xrightarrow{\widehat e_4}\widehat v_1$
   & 12 & $\mathsf I$
       & $\widehat v_5\xrightarrow{\widehat e_6}\widehat v_1$
\end{tabular}
\end{center}

Concatenating these steps gives the decoded rotated walk
$\widehat W_{\mathrm{rot}}$. The two $\mathsf{close}$ records are
exactly the $g=2$ first traversals of core edges that arrive at already
exposed core vertices.

\medskip
\textbf{4. Undo the rotation and canonicalize.}
Undoing $r=1$ moves Step~12 to the front and makes $\widehat v_5$ the
starting vertex. The first-encounter canonical relabeling is
\[
(\widehat v_5,\widehat v_1,\widehat v_2,\widehat v_3,\widehat v_4)
\mapsto
(v_1,v_2,v_3,v_4,v_5),
\]
\[
(\widehat e_6,\widehat e_1,\widehat e_2,\widehat e_3,\widehat e_4,\widehat e_5)
\mapsto
(e_1,e_2,e_3,e_4,e_5,e_6).
\]
This recovers exactly the canonical walk from Step~1. The three core
edges contribute the factor $x^{e_c}=x^3$, and the three
$\mathsf{new}$ forest records contribute another factor $x^3$. The total weight is
\[
    x^3\cdot x^3=x^6=x^{e(W)}.
\]
\end{minipage}}
\caption{A complete encoding--decoding example.}
\label{fig:encoding-example}
\end{figure}

\subsection{Using the cycle-rank penalty}
Let
\[
    C_3=18,
    \qquad
    C_4=\log16.
\]

\begin{claim}\label{clm:positive-rank-sum}
For every integer $\ell\ge2$, every $x\ge\ell$, and every $d\ge16^\ell$,
\[
    \sum_{g\ge1}d^{-g}A_{\ell,g}(x)
    \le\frac23(18x)^\ell.
\]
\end{claim}

\begin{proof}
For $g=1$, \eqref{eq:unicyclic} gives
\[
    \frac{d^{-1}A_{\ell,1}(x)}{(18x)^\ell}
    \le\frac{2\ell^2}{4^\ell}
    \le\frac12.
\]

For $g\ge2$, the loopless core has at least two vertices, and hence
$g\le e_c-1\le\ell-1$. Define
\[
    a_g:=\frac{4(g+1)^2}{16^g}.
\]
By \eqref{eq:Apg-core},
\begin{equation}\label{eq:normalized-rank-g}
    \frac{d^{-g}A_{\ell,g}(x)}{(18x)^\ell}
    \le2\ell^{g+1}a_g^\ell.
\end{equation}
Moreover,
\[
    \frac{a_{g+1}}{a_g}
    =\frac1{16}\left(\frac{g+2}{g+1}\right)^2
    \le\frac19
    \qquad(g\ge2),
\]
so $a_g\le(9/64)9^{-(g-2)}$. Since $\ell\ge3$ whenever $g\ge2$,
\begin{align*}
    \sum_{g=2}^{\ell-1}2\ell^{g+1}a_g^\ell
    &\le
    2\ell^3\left(\frac9{64}\right)^\ell
    \sum_{k\ge0}\left(\frac{\ell}{9^\ell}\right)^k \\
    &\le
    \frac{54(9/64)^3}{1-1/243}
    <\frac16.
\end{align*}
Here we used that $2\ell^3(9/64)^\ell$ is decreasing for $\ell\ge3$ and
that $\ell/9^\ell\le1/243$. Combining the two ranges proves the claim.
\end{proof}

Under the condition $\ell\le\log d/C_4$, we have $d\ge16^\ell$. Claim~\ref{clm:positive-rank-sum} and the tree bound \eqref{eq:tree-count} give
\[
\begin{aligned}
    \sum_{W\in\cW_\ell}d^{-g(W)}x^{e(W)}
    &\le(8x)^\ell+\frac23(18x)^\ell\\
    &\le(18x)^\ell,
\end{aligned}
\]
since $(8/18)^\ell\le(4/9)^2<1/3$ for $\ell\ge2$. This proves Lemma~\ref{lem:walk-count}.

\section{Controlling lower distortion}\label{sec:shrinkage}

Sections~\ref{sec:shrinkage}--\ref{sec:expansion-recursion} prove
Proposition~\ref{prop:mean-to-moments}. Recall that
\[
    Y=(-\lambda_{\min}(E))_+,
    \qquad
    Z=(\lambda_{\max}(E))_+,
    \qquad
    R:=\lVert E\rVert_{\op}=\max\{Y,Z\},
    \qquad
    \mu:=\E R.
\]
We begin with the lower distortion \(Y\). Its layerwise contributions
are independent and centered, and the positive semidefiniteness of each
layer matrix gives the one-sided bound needed for Bousquet's
inequality. \space
The resulting bound depends on the mean only through $\mu$.

\begin{lemma}[Lower distortion moments]\label{lem:shrinkage}
For every $q\ge2$,
\begin{equation}\label{eq:shrinkage-moment}
    \lVert Y\rVert_q
    \le \E Y
      +O\!\left(
          \sqrt{\frac{q}{m}}
          +\sqrt{\frac{q\,\E Y}{s}}
          +\frac{q}{s}
        \right).
\end{equation}
Consequently,
\begin{equation}\label{eq:shrinkage-simplified}
    \lVert Y\rVert_q
    =O\!\left(
        \mu+\sqrt{\frac qm}+\frac qs
    \right).
\end{equation}
\end{lemma}

\begin{proof}
Let $S_\gamma\in\R^{B\times n}$ be the block of $S$ corresponding to
layer $\gamma$, and let
\begin{equation*}
    T_\gamma=U^\top S_\gamma^\top S_\gamma U.
\end{equation*}

Then
\begin{equation*}
    \Id+E=\frac1s\sum_{\gamma=1}^s T_\gamma,
    \qquad
    \E T_\gamma=\Id,
    \qquad
    T_\gamma\succeq0.
\end{equation*}

\paragraph{Step 1: express lower distortion as a supremum of layerwise deviations.}
For \(\lVert x\rVert_2=1\), define
\begin{equation*}
    L_{\gamma,x}
    :=
    \frac{1}{s}\bigl(1-x^\top T_\gamma x\bigr).
\end{equation*}
Since
\[
    \Id+E=\frac{1}{s}\sum_{\gamma=1}^s T_\gamma,
\]
we have, for every \(\lVert x\rVert_2=1\),
\[
    -x^\top E x
    =
    \sum_{\gamma=1}^s L_{\gamma,x}.
\]

Introduce an auxiliary index \(\ast\) and set
\[
    L_{\gamma,\ast}:=0
    \qquad (\gamma\in[s]).
\]
With
\[
    \mathcal X
    :=
    \bigl\{x\in\mathbb R^d:\lVert x\rVert_2=1\bigr\}
    \cup\{\ast\},
\]
we may write
\begin{equation}\label{eq:shrinkage-supremum}
    Y
    =
    \sup_{x\in\mathcal X}
    \sum_{\gamma=1}^s L_{\gamma,x}.
\end{equation}
The same augmented index set \(\mathcal X\) will be used in the
weak-variance calculation in Step~2 and in the concentration argument
in Step~3.

For each fixed \(x\in\mathcal X\), the random variables
\(L_{\gamma,x}\) are independent across \(\gamma\). They are centered: this is immediate when \(x=\ast\), and for
\(\lVert x\rVert_2=1\),
\[
    \E L_{\gamma,x}
    =
    \frac{1}{s}\bigl(1-x^\top \E T_\gamma x\bigr)
    =
    0
\]
because \(\E T_\gamma=\Id\). Moreover, when
\(\lVert x\rVert_2=1\), the positivity \(T_\gamma\succeq0\) gives
\(x^\top T_\gamma x\ge0\), and hence
\begin{equation*}
    L_{\gamma,x}\le \frac{1}{s}.
\end{equation*}
Thus the indexed deviations are centered and bounded above by \(1/s\), as required for
the one-sided concentration argument in Step~3.

\paragraph{Step 2: bound the variance in each fixed direction.}
Fix \(x\) with \(\lVert x\rVert_2=1\), and write
\[
    a_i:=\langle u_i,x\rangle.
\]
The diagonal terms in \(T_\gamma\) sum to
\(\sum_i u_i u_i^\top=\Id\). Hence
\[
    x^\top(T_\gamma-\Id)x
    =
    \sum_{i\ne j}
    \mathbf 1\{h_\gamma(i)=h_\gamma(j)\}
    \sigma_{\gamma i}\sigma_{\gamma j}a_i a_j
    =
    2\sum_{i<j}
    \mathbf 1\{h_\gamma(i)=h_\gamma(j)\}
    \sigma_{\gamma i}\sigma_{\gamma j}a_i a_j.
\]
This random variable has expectation 0. Moreover, after squaring and averaging
over the Rademacher signs, all mixed terms vanish unless they correspond
to the same unordered pair \(\{i,j\}\). Since two distinct coordinates
collide with probability \(1/B\),
\begin{align*}
    \E\bigl[x^\top(T_\gamma-\Id)x\bigr]^2
    &=
    \frac{4}{B}\sum_{i<j}a_i^2a_j^2 \\
    &\le
    \frac{2}{B}\left(\sum_i a_i^2\right)^2
    =
    \frac{2}{B}.
\end{align*}
Here the last equality follows from
\[
    \sum_i a_i^2
    =
    \sum_i\langle u_i,x\rangle^2
    =
    x^\top\left(\sum_i u_i u_i^\top\right)x
    =
    1.
\]

Since
\[
    L_{\gamma,x}
    =
    -\frac{1}{s}x^\top(T_\gamma-\Id)x,
\]
we conclude that
\[
    \E L_{\gamma,x}^2
    \le
    \frac{2}{s^2B}.
\]
Summing over the \(s\) layers and taking the supremum over
\(\lVert x\rVert_2=1\) gives
\begin{equation*}
    \sup_{\lVert x\rVert_2=1}
    \sum_{\gamma=1}^s \E L_{\gamma,x}^2
    \le
    \frac{2}{sB}
    =
    \frac{2}{m}.
\end{equation*}
The auxiliary zero index introduced in Step~1 has no contribution to this variance parameter.

\paragraph{Step 3: apply Bousquet's inequality.}
Since the map
\[
    x\longmapsto\sum_{\gamma=1}^s L_{\gamma,x}
\]
is continuous on the unit sphere for every realization of the sketch,
the supremum in \eqref{eq:shrinkage-supremum} is unchanged if the sphere
is replaced by a fixed countable dense subset. Therefore the countability
hypothesis in Lemma~\ref{lem:bousquet} is satisfied.

Regard the complete randomness in layer $\gamma$ as the $\gamma$th
observation. These observations are independent and identically distributed,
and each $L_{\gamma,x}$ is obtained by applying the same measurable function
to the $\gamma$th observation.
Bousquet's inequality, stated in Lemma~\ref{lem:bousquet}, applied with upper bound $1/s$ and variance $2/m$, gives for every $t\ge0$,
\begin{equation*}
    \Pp\!\left\{
        Y>\E Y
          +2\sqrt{\left(\frac1m+\frac{\E Y}{s}\right)t}
          +\frac{t}{3s}
    \right\}
    \le e^{-t}.
\end{equation*}
Using $\sqrt{a+b}\le\sqrt a+\sqrt b$ and the tail-to-moment conversion in Lemma~\ref{lem:tail-to-moment}, we obtain for every $q\ge2$,
\[
    \lVert Y\rVert_q
    \le \E Y
      +c\left(
          \sqrt{\frac{q}{m}}
          +\sqrt{\frac{q\,\E Y}{s}}
          +\frac{q}{s}
        \right),
\]
which is \eqref{eq:shrinkage-moment}.

Finally, $\E Y\le\mu$. Substituting this into
\eqref{eq:shrinkage-moment} and using
\[
    2\sqrt{\frac{\mu q}{s}}
    \le \mu+\frac qs
\]
proves \eqref{eq:shrinkage-simplified}.
\end{proof}

\section{Controlling upper distortion}
\label{sec:expansion-sensitivity}

We now turn to the upper distortion
\[
    Z=(\lambda_{\max}(E))_+.
\]
Unlike the lower distortion, the contribution of a single layer to
\(Z\) has no useful deterministic upper bound. We therefore resample
one individual random choice at a time and control the resulting
one-sided variance proxy using polynomial Efron--Stein.

\subsection{Coordinate resampling}
\label{sec:resampling}

The independent random variables are
\[
    X_\alpha=(h_\gamma(i),\sigma_{\gamma i}),
    \qquad
    \alpha=(\gamma,i)\in[s]\times[n].
\]
Let \(X_\alpha'\) be an independent copy of \(X_\alpha\), and let
\(E^{(\alpha)}\) and
\[
    Z^{(\alpha)}
    =
    \bigl(\lambda_{\max}(E^{(\alpha)})\bigr)_+
\]
denote the error matrix and upper distortion after replacing
\(X_\alpha\) by \(X_\alpha'\). Writing \(\E_\alpha'\) for expectation
over the resampled copy conditional on the original inputs, define
\[
    \mathcal V^+
    =
    \sum_\alpha
    \E_\alpha'\!\left[
        \bigl(Z-Z^{(\alpha)}\bigr)_+^2
    \right].
\]
Polynomial Efron--Stein reduces the upper-distortion bound to controlling
\(\mathcal V^+\). We will prove
\[
    \mathcal V^+
    =
    O\!\left(
        \mathcal K+\mathcal Q+\frac{1+Z}{m}
    \right),
\]
where \(\mathcal K\) and \(\mathcal Q\) are introduced in
Section~\ref{sec:old-bucket}.

If \(Z=0\), then
\(Z^{(\alpha)}\ge0\) for every \(\alpha\), and hence
\(\mathcal V^+=0\). We may therefore assume \(Z>0\) and choose a unit
top eigenvector \(x\) of \(E\) satisfying
\[
    x^\top Ex=Z.
\]

For a unit vector \(v\), a layer \(\gamma\), and a coordinate \(j\),
define
\[
    t_{\gamma j}(v)
    =
    \frac{\sigma_{\gamma j}\langle u_j,v\rangle}{\sqrt{s}}.
\]
For a layer--bucket pair \(\beta=(\gamma,b)\), define its bucket load by
\[
    A_\beta(v)
    =
    \sum_{j\in I_\beta}t_{\gamma j}(v).
\]
Thus \(A_\beta(v)\) is the total signed contribution, in direction
\(v\), of the coordinates assigned to bucket \(\beta\). In the
resampling calculation, \(v=x\) is the fixed top eigenvector chosen
above, and we suppress the argument \(x\). The bucket loads give a scalar representation of the Rayleigh quotient:
\[
    v^\top(\Id+E)v
    =
    \lVert\Phi Uv\rVert_2^2
    =
    \sum_\beta A_\beta(v)^2.
\]
Thus resampling $X_\alpha$ changes the Rayleigh
quotient only through the bucket loads affected by removing and
reinserting that coordinate.

Fix $\alpha=(\gamma,i)$. Resampling $X_\alpha$ changes only two bucket loads:
the load of the original bucket, from which the coordinate is removed,
and the load of the resampled bucket, to which it is added. Evaluating
this change in the fixed direction \(x\) will split
\(\mathcal V^+\) into the corresponding original- and
resampled-bucket contributions. For the calculation below, abbreviate
\[
    h_i=h_\gamma(i),
    \qquad
    h_i'=h_\gamma'(i),
    \qquad
    t_i=t_{\gamma i},
    \qquad
    t_i'
    =\frac{\sigma_{\gamma i}'\langle u_i,x\rangle}{\sqrt{s}},
\]
and, for $b\in[B]$, define the leave-one-out load
\[
    A_{\gamma,b}^{(-i)}
    :=A_{(\gamma,b)}-t_i\one\{h_i=b\}.
\]
Thus $A_{\gamma,b}^{(-i)}$ is the load in bucket $b$ after coordinate $i$ has been removed.
Before resampling, the bucket loads in layer $\gamma$ are
\[
    A_{\gamma,b}^{(-i)}+t_i\one\{h_i=b\},
\]
whereas after resampling they are
\[
    A_{\gamma,b}^{(-i)}+t_i'\one\{h_i'=b\}.
\]
Only layer $\gamma$ changes, and $(t_i')^2=t_i^2$. Hence
\begin{align*}
    x^\top(E-E^{(\alpha)})x
    &=\sum_{b=1}^B
      \left[
        \bigl(A_{\gamma,b}^{(-i)}+t_i\one\{h_i=b\}\bigr)^2
        -\bigl(A_{\gamma,b}^{(-i)}+t_i'\one\{h_i'=b\}\bigr)^2
      \right]\\
    &=2\bigl(t_i A_{\gamma,h_i}^{(-i)}-t_i'A_{\gamma,h_i'}^{(-i)}\bigr).
\end{align*}
By the Rayleigh--Ritz variational principle,
\[
    Z^{(\alpha)}\ge x^\top E^{(\alpha)}x,
\]
so
\begin{equation*}
    Z-Z^{(\alpha)}
    \le 2\bigl(t_i A_{\gamma,h_i}^{(-i)}-t_i'A_{\gamma,h_i'}^{(-i)}\bigr).
\end{equation*}

For $z\in\R$, write $(z)_-=\max\{-z,0\}$. Since
\[
    (u-v)_+^2\le2(u)_+^2+2(v)_-^2,
\]
we obtain
\begin{equation}\label{eq:variance-split}
    \mathcal V^+
    \le8J_{\mathrm{original}}+8J_{\mathrm{resampled}},
\end{equation}
where
\begin{equation*}
    J_{\mathrm{original}}
    =\sum_\alpha\bigl(t_i A_{\gamma,h_i}^{(-i)}\bigr)_+^2,
    \qquad
    J_{\mathrm{resampled}}
    =\sum_\alpha
      \E_\alpha'\!\left[\bigl(t_i'A_{\gamma,h_i'}^{(-i)}\bigr)_-^2\right].
\end{equation*}
We now bound these two terms separately.

\subsection{The resampled bucket}
\label{sec:fresh-bucket}

Fix $\alpha=(\gamma,i)$ and condition on the current sketch. The leave-one-out
loads $A_{\gamma,1}^{(-i)},\ldots,A_{\gamma,B}^{(-i)}$ are then fixed,
while $h_i'$ is uniform on $[B]$ and
$\sigma_{\gamma i}'$ is an independent Rademacher sign. Averaging first over
the resampled sign and then over the resampled bucket gives
\begin{equation}\label{eq:fresh-average}
    \E_\alpha'\!\left[\bigl(t_i'A_{\gamma,h_i'}^{(-i)}\bigr)_-^2\right]
    =\frac{\langle u_i,x\rangle^2}{2sB}
      \sum_{b=1}^B \bigl(A_{\gamma,b}^{(-i)}\bigr)^2
    =\frac{\langle u_i,x\rangle^2}{2sB}
      \sum_{b=1}^B
      \left(
        A_{(\gamma,b)}
        -t_{\gamma i}\one\{i\in I_{(\gamma,b)}\}
      \right)^2.
\end{equation}
Note that the factor $1/B$ comes from averaging over the resampled bucket; it would be absent if only the sign were resampled.

For each fixed layer $\gamma$, the inequality $(u-v)^2\le2u^2+2v^2$ gives
\begin{align}
    &\sum_{i=1}^n\langle u_i,x\rangle^2
      \sum_{b=1}^B
      \left(
        A_{(\gamma,b)}
        -t_{\gamma i}\one\{i\in I_{(\gamma,b)}\}
      \right)^2
      \nonumber\\
    &\qquad\le
      2\sum_{i=1}^n\langle u_i,x\rangle^2
      \sum_{b=1}^B A_{(\gamma,b)}^2
      +2\sum_{i=1}^n\langle u_i,x\rangle^2t_{\gamma i}^2
      \nonumber\\
    &\qquad\le
      2\sum_{b=1}^B A_{(\gamma,b)}^2+\frac2s.
    \label{eq:fresh-sum-bound}
\end{align}
Here we used
\[
    \sum_{i=1}^n\langle u_i,x\rangle^2=1,
    \qquad
    \sum_{i=1}^n\langle u_i,x\rangle^4
    \le
    \left(\sum_{i=1}^n\langle u_i,x\rangle^2\right)^2
    =1.
\]
Summing \eqref{eq:fresh-average} over $\gamma$ and $i$, applying
\eqref{eq:fresh-sum-bound}, and using $m=sB$, we obtain
\begin{align}
    J_{\mathrm{resampled}}
    &\le
      \frac1m\left(
        \sum_\beta A_\beta^2+1
      \right)\nonumber\\
    &=\frac{\lVert\Phi Ux\rVert_2^2+1}{m}
      =\frac{Z+2}{m}.
    \label{eq:fresh-final}
\end{align}

\subsection{The original bucket}
\label{sec:old-bucket}

We next control the original-bucket term $J_{\mathrm{original}}$ bucket by bucket. For a unit vector $v$
and a layer--bucket pair $\beta=(\gamma,b)$, define
\begin{equation*}
    Q_\beta(v)
    =A_\beta(v)^2
      -\sum_{i\in I_\beta}t_{\gamma i}(v)^2,
\end{equation*}
and
\begin{equation*}
    K_\beta(v)
    =\left(\sum_{i\in I_\beta}t_{\gamma i}(v)^2\right)^2
      -\sum_{i\in I_\beta}t_{\gamma i}(v)^4.
\end{equation*}
The quantity $Q_\beta(v)$ is the signed bucket excess in bucket $\beta$: it is
the amount by which the square of the total bucket load exceeds the sum of the
individual squared loads. Thus $Q_\beta(v)$ is the bucket's signed quadratic
contribution to the Rayleigh error. The quantity $K_\beta(v)$ is the
corresponding interaction energy: it squares each pairwise interaction before
summing, so it is nonnegative and does not allow sign cancellation. Indeed,
\[
    Q_\beta(v)
    =2\sum_{\substack{i<j\\i,j\in I_\beta}}
      t_{\gamma i}(v)t_{\gamma j}(v),
    \qquad
    K_\beta(v)
    =2\sum_{\substack{i<j\\i,j\in I_\beta}}
      t_{\gamma i}(v)^2t_{\gamma j}(v)^2.
\]
Moreover,
\begin{equation}\label{eq:bucket-basic}
    0\le\sum_{i\in I_\beta}t_{\gamma i}(v)^2\le\frac1s,
    \qquad
    \sum_\beta\sum_{i\in I_\beta}t_{\gamma i}(v)^2=1,
    \qquad
    v^\top Ev=\sum_\beta Q_\beta(v).
\end{equation}
The first two statements follow from the Parseval identity
$\sum_i\langle u_i,v\rangle^2=1$, and the last follows from
\[
    \sum_\beta A_\beta(v)^2
    =\lVert\Phi Uv\rVert_2^2
    =1+v^\top Ev.
\]

We now specialize these quantities to the fixed maximizing direction $x$ and
again suppress the argument. Since the leave-one-out load in coordinate $i$'s
current bucket is
\[
    A_{\gamma,h_i}^{(-i)}=A_{(\gamma,h_\gamma(i))}-t_{\gamma i},
\]
the original-bucket contribution is exactly
\begin{equation*}
    J_{\mathrm{original}}
    =\sum_\beta\sum_{i\in I_\beta}
      \bigl(t_{\gamma i}(A_\beta-t_{\gamma i})\bigr)_+^2.
\end{equation*}
The signed bucket excess $Q_\beta$ alone cannot control this sum, because
positive and negative coordinate influences may cancel. The algebraic lemma~\ref{lem:bucket-influence} therefore can be used to combine the interaction energy $K_\beta$ with a truncated bucket excess. Define
\begin{equation*}
    \psi(z)
    =\min\left\{(z)_+^2,\frac{(z)_+}{s}\right\}.
\end{equation*}
The two properties needed in Section~\ref{sec:collision-truncation} are
\begin{equation}\label{eq:psi-properties}
    \psi(0)=0,
    \qquad
    0\le\psi(z)\le z^2,
    \qquad
    |\psi(z)-\psi(z')|\le\frac2s|z-z'|.
\end{equation}
The quadratic bound will control the conditional mean of the truncated bucket excess, while the Lipschitz bound will allow Rademacher contraction.

Define the global quantities
\begin{equation*}
    \mathcal{K}
    =\sup_{\lVert v\rVert_2=1}\sum_\beta K_\beta(v),
    \qquad
    \mathcal{Q}
    =\sup_{\lVert v\rVert_2=1}
      \sum_\beta\psi\bigl(Q_\beta(v)\bigr).
\end{equation*}

\begin{lemma}[Upper-distortion variance proxy]\label{lem:expansion-variance-proxy}
The variance proxy satisfies
\begin{equation}\label{eq:structural-variance-proxy}
    \mathcal V^+
    =O\!\left(
        \mathcal{K}
        +\mathcal{Q}
        +\frac{1+Z}{m}
    \right).
\end{equation}
\end{lemma}

\begin{proof}
Since $\sum_{i\in I_\beta}t_{\gamma i}^2\le1/s$, the algebraic influence lemma from
Appendix~\ref{app:bucket-influence} gives, for every bucket $\beta$,
\begin{equation}\label{eq:bucket-influence-application}
    \sum_{i\in I_\beta}
      \bigl(t_{\gamma i}(A_\beta-t_{\gamma i})\bigr)_+^2
    \le c \bigl(K_\beta+\psi(Q_\beta)\bigr).
\end{equation}

$\mathcal{K}$ aggregates the nonnegative interaction energies
$K_\beta$, while $\mathcal{Q}$ aggregates the truncated positive bucket
excesses $\psi(Q_\beta)$. Because the selected $x$ is one admissible unit
vector, \eqref{eq:bucket-influence-application} yields
\begin{equation}\label{eq:J-old-control}
    J_{\mathrm{original}}
    \le  c \bigl(\mathcal{K}+\mathcal{Q}\bigr).
\end{equation}

Combining \eqref{eq:variance-split}, \eqref{eq:fresh-final}, and
\eqref{eq:J-old-control} gives \eqref{eq:structural-variance-proxy}.
\end{proof}

\section{Interaction energy and truncated bucket excess bounds}
\label{sec:collision-truncation}

Recall the original-bucket quantities
\[
    \mathcal K
    =
    \sup_{\lVert v\rVert_2=1}\sum_\beta K_\beta(v),
    \qquad
    \mathcal Q
    =
    \sup_{\lVert v\rVert_2=1}
    \sum_\beta\psi\bigl(Q_\beta(v)\bigr).
\]
We first bound the global interaction energy \(\mathcal K\) using matrix
Bernstein. We then control the global truncated bucket excess \(\mathcal Q\) by
conditioning on the hashes and applying Rademacher contraction.

\subsection{Interaction-energy moments}

\begin{lemma}[Interaction energy bound]\label{lem:collision-bound}
For every $q\ge2\log(2d^2)$,
\begin{equation*}
    \lVert\mathcal{K}\rVert_q
    =O\!\left(\frac1m+\frac{q}{s^2}\right).
\end{equation*}
\end{lemma}

\begin{proof}
We represent $\mathcal{K}$ through a sum of independent positive semidefinite matrices and apply matrix Bernstein. For $i\ne j$, let
\[
    \xi_{ij}=u_i\otimes u_j\in\R^{d^2},
\]
where $\otimes$ denotes the Kronecker product. Define
\begin{equation*}
    M_\gamma
    =\sum_{i\ne j}
      \one\{h_\gamma(i)=h_\gamma(j)\}
      \xi_{ij}\xi_{ij}^\top.
\end{equation*}
Because
\[
    \sum_{i,j}\xi_{ij}\xi_{ij}^\top
    =
    \left(\sum_i u_i u_i^\top\right)
    \otimes
    \left(\sum_j u_j u_j^\top\right)
    =I_{d^2},
\]
we have
\begin{equation*}
    0\preceq M_\gamma\preceq I_{d^2},
    \qquad
    \E M_\gamma
    =\frac1B\sum_{i\ne j}\xi_{ij}\xi_{ij}^\top
    \preceq\frac1B I_{d^2}.
\end{equation*}
Furthermore,
\begin{equation*}
    \sum_\beta K_\beta(x)
    =\frac1{s^2}
      \left\langle
        x\otimes x,\,
        \left(\sum_{\gamma=1}^sM_\gamma\right)(x\otimes x)
      \right\rangle,
\end{equation*}
so
\begin{equation*}
    \mathcal{K}
    \le\frac1{s^2}
      \left\lVert\sum_{\gamma=1}^sM_\gamma\right\rVert_{\op}.
\end{equation*}
The result follows from the matrix Bernstein moment corollary in Lemma~\ref{lem:matrix-bernstein-moment}.
\end{proof}

\subsection{A Lipschitz reduction of the truncated bucket excess}

We next control $\mathcal{Q}$. Its mean part is controlled by
$\mathcal{K}$, while its centered part is reduced to the underlying
operator-norm process.

\begin{lemma}[Truncated bucket excess reduction]\label{lem:truncated-bucket}
For every $q\ge1$,
\begin{equation*}
    \lVert\mathcal{Q}\rVert_q
    =O\!\left(
        \lVert\mathcal{K}\rVert_q
        +\frac1s\bigl\lVert\lVert E\rVert_{\op}\bigr\rVert_q
    \right).
\end{equation*}
\end{lemma}

\begin{proof}
Condition on all hashes. Different buckets then use disjoint collections of Rademacher signs, so the random signed-bucket-excess functions $Q_\beta(\cdot)$ are independent across buckets. They are centered, and for fixed $x$,
\begin{equation}\label{eq:bucket-error-second-moment}
    \E_\sigma Q_\beta(x)^2
    =4\sum_{\substack{i<j\\i,j\in I_\beta}}
      t_{\gamma i}^2t_{\gamma j}^2
    =2K_\beta(x).
\end{equation}
Since $\psi(z)\le z^2$,
\begin{equation}\label{eq:trunc-mean}
    \sup_{\lVert x\rVert_2=1}
    \sum_\beta\E_\sigma\psi(Q_\beta(x))
    \le2\mathcal{K}.
\end{equation}
Decompose $\mathcal{Q}$ into this mean part and the supremum of the centered process. Apply the centered-function contraction lemma, Lemma~\ref{lem:centered-contraction}, with Lipschitz constant $2/s$. Since
\[
    \sum_\beta Q_\beta(x)=x^\top Ex,
\]
for a random variable $X$, write
$\lVert X\rVert_{L_q(\sigma\mid h)}=(\E_\sigma[|X|^q\mid h])^{1/q}$.
We obtain, conditionally on the hashes,
\[
\left\lVert
    \sup_{\lVert x\rVert_2=1}
    \left|
      \sum_\beta
      \bigl(
        \psi(Q_\beta(x))
        -\E_\sigma\psi(Q_\beta(x))
      \bigr)
    \right|
\right\rVert_{L_q(\sigma\mid h)}
\le
\frac{c}{s}
\bigl\lVert\lVert E\rVert_{\op}\bigr\rVert_{L_q(\sigma\mid h)}.
\]
Combine this with \eqref{eq:trunc-mean}, then take the $L_q$ norm over the hashes.
\end{proof}

The inequality $\psi(z)\le z^2$ controls the conditional mean through \eqref{eq:bucket-error-second-moment}, while the $2/s$-Lipschitz bound makes the centered process comparable to $\lVert E\rVert_{\op}$ without a second net or trace-moment argument.

\section{Closing the operator-norm moment bound}
\label{sec:expansion-recursion}

Write
\[
    r_q
    =
    \bigl\lVert\lVert E\rVert_{\op}\bigr\rVert_q,
    \qquad
    z_q=\lVert Z\rVert_q,
    \qquad
    y_q=\lVert Y\rVert_q.
\]
Lemma~\ref{lem:shrinkage} already gives the required bound for \(y_q\).
We now combine the variance-proxy estimate from
Section~\ref{sec:expansion-sensitivity} with the bounds for
\(\mathcal K\) and \(\mathcal Q\) from
Section~\ref{sec:collision-truncation}. This yields a recursive bound
for \(z_q\) in terms of \(r_q\). We then combine the two distortions and
absorb the recursive term to prove
Proposition~\ref{prop:mean-to-moments}.

\begin{lemma}[Upper distortion recursion]\label{lem:expansion-recursion}
For every $q\ge4\log(2d^2)$,
\begin{equation}\label{eq:expansion-recursion}
    \lVert Z\rVert_q
    =O\!\left(
        \mu
        +\left[
            \frac qm
            +\frac{q^2}{s^2}
            +\frac qs\,\bigl\lVert\lVert E\rVert_{\op}\bigr\rVert_q
        \right]^{1/2}
    \right).
\end{equation}
\end{lemma}

\begin{proof}
By Lemma~\ref{lem:expansion-variance-proxy}, Minkowski's inequality,
Lemma~\ref{lem:collision-bound}, and
Lemma~\ref{lem:truncated-bucket}, we obtain
\begin{equation}\label{eq:variance-proxy-moment}
    \lVert\mathcal V^+\rVert_{q/2}
    \le
    c\left[
        \frac1m
        +\frac q{s^2}
        +\frac{r_q}{s}
        +\frac{1+r_q}{m}
    \right].
\end{equation}
Indeed, monotonicity of $L_q$ norms gives
\[
    \bigl\lVert\lVert E\rVert_{\op}\bigr\rVert_{q/2}
    \le r_q,
\]
and
\[
    \lVert 1+Z\rVert_{q/2}
    \le 1+\lVert Z\rVert_{q/2}
    \le 1+r_q.
\]
Thus no decomposition of $\lVert E\rVert_{\op}$ into $Y+Z$ is needed
in the variance-proxy bound.

Since
\[
    Z\le\E Z+(Z-\E Z)_+,
\]
polynomial Efron--Stein inequality in Lemma~\ref{lem:polynomial-efron-stein} gives
\[
    z_q
    \le
    \E Z
    +c\sqrt{q\,\lVert\mathcal V^+\rVert_{q/2}}.
\]
Moreover, $\E Z\le\E R=\mu$. Substituting
\eqref{eq:variance-proxy-moment} yields
\[
    z_q
    \le
    \mu
    +c\left[
        \frac qm
        +\frac{q^2}{s^2}
        +\frac qs\,r_q
        +\frac qm(1+r_q)
    \right]^{1/2}.
\]
Finally, $m=sB\ge s$, so
\[
    \frac qm\,r_q\le\frac qs\,r_q
    \qquad\text{and}\qquad
    \frac qm\le\frac qm+\frac{q^2}{s^2}.
\]
This proves \eqref{eq:expansion-recursion}.
\end{proof}

\begin{proof}[Proof of Proposition~\ref{prop:mean-to-moments}]
Fix $q\ge4\log(2d^2)$ and set
\[
    a_q
    =
    \mu+\sqrt{\frac qm}+\frac qs.
\]
Lemma~\ref{lem:shrinkage} gives
\begin{equation}\label{eq:y-by-aq}
    y_q\le ca_q.
\end{equation}
Lemma~\ref{lem:expansion-recursion} and
$\sqrt{a+b+c}\le\sqrt a+\sqrt b+\sqrt c$ give
\begin{equation}\label{eq:z-by-rq}
    z_q
    \le
    ca_q
    +c\sqrt{\frac qs\,r_q}.
\end{equation}

Because $R=\max\{Y,Z\}$,
\[
    R^q\le Y^q+Z^q.
\]
Therefore
\[
    r_q
    \le
    \bigl(y_q^q+z_q^q\bigr)^{1/q}
    \le y_q+z_q.
\]
Combining this with \eqref{eq:y-by-aq} and \eqref{eq:z-by-rq},
\[
    r_q
    \le
    2ca_q
    +c\sqrt{\frac qs\,r_q}.
\]
Young's inequality gives
\[
    c\sqrt{\frac qs\,r_q}
    \le
    \frac12r_q+\frac{c^2}{2}\frac qs
    \le
    \frac12r_q+\frac{c^2}{2}a_q.
\]
Absorbing $\frac12r_q$ into the left-hand side gives
\begin{equation}\label{eq:two-sided-moment-local}
    r_q
    \le
    (4c+c^2)\left(
        \mu+\sqrt{\frac qm}+\frac qs
    \right).
\end{equation}

It remains to derive the tail statement. For $t\ge0$, set
\[
    q_t=4\log(2d^2)+t.
\]
Markov's inequality gives
\[
    \Pp\{R>e\lVert R\rVert_{q_t}\}
    \le e^{-q_t}
    \le e^{-t}.
\]
Moreover,
\[
    q_t\le8\bigl(\log(2d)+t\bigr).
\]
Applying \eqref{eq:two-sided-moment-local} at $q=q_t$ therefore gives
\[
    e\lVert R\rVert_{q_t}
    \le
    8e(4c+c^2)\left(
        \mu
        +\sqrt{\frac{\log(2d)+t}{m}}
        +\frac{\log(2d)+t}{s}
    \right).
\]
Choose
\[
    C_{\mathrm{conc}}\ge8e(4c+c^2).
\]
Then \eqref{eq:two-sided-moment-local} implies
\eqref{eq:mean-to-moments}, and the last two inequalities imply
\eqref{eq:mean-to-tail}.
\end{proof}

\section{Proof of the main theorem}\label{sec:main-proof}

We now choose the SparseStack parameters and apply
Propositions~\ref{prop:mean-error} and~\ref{prop:mean-to-moments}. All parameters below depend
only on \(d,\eps,\delta\), and not on the target subspace, as required for
an OSE.

\begin{proof}[Proof of Theorem~\ref{thm:main}]
Let $C_1$ and $C_{\mathrm{conc}}$ be the universal constants from
Propositions~\ref{prop:mean-error} and~\ref{prop:mean-to-moments}, respectively, and set
\[
    C_0=\max\{4,C_{\mathrm{conc}}(C_1+2)\},
    \qquad
    \eps_0=\frac{\eps}{C_0}.
\]
Choose
\[
    p=8\left\lceil\log\frac{8d}{\delta}\right\rceil,
    \qquad
    s=\left\lceil\frac{p}{8\eps_0}\right\rceil,
    \qquad
    B=\left\lceil
        \frac{8\bigl(d+\log(1/\delta)\bigr)}{\eps_0p}
      \right\rceil,
    \qquad
    m=sB.
\]

We first verify the hypotheses of
Proposition~\ref{prop:mean-error}. The definition of \(p\)
gives \(p\ge\log d\), and
\[
    p
    \le
    8\left(\log\frac{8d}{\delta}+1\right)
    \le
    16\bigl(d+\log(1/\delta)\bigr).
\]
Since \(\eps_0\le1/4\), we have
\[
    \frac{p}{8\eps_0}\ge1,
    \qquad
    \frac{8\bigl(d+\log(1/\delta)\bigr)}{\eps_0p}
    \ge\frac{1}{2\eps_0}
    \ge2.
\]
Each ceiling is at most twice its argument, and hence
\[
    \frac{p}{8\eps_0}
    \le s
    \le\frac{p}{4\eps_0},
    \qquad
    \frac{8\bigl(d+\log(1/\delta)\bigr)}{\eps_0p}
    \le B
    \le
    \frac{16\bigl(d+\log(1/\delta)\bigr)}{\eps_0p}.
\]
In particular,
\[
    \frac{d+\log(1/\delta)}{\eps_0^2}
    \le m
    \le
    \frac{4\bigl(d+\log(1/\delta)\bigr)}{\eps_0^2}.
\]
The lower bounds on \(s\) and \(B\), together with
\(d\le d+\log(1/\delta)\), verify the assumptions in
\eqref{eq:technical-regime} with \(\eps\) replaced by \(\eps_0\).

Now set
\[
    t=\log(1/\delta).
\]
Since \(d\ge2\),
\[
    \log(2d)+t
    =\log\frac{2d}{\delta}
    \le d+\log(1/\delta).
\]
The lower bound on \(m\) therefore gives
\[
    \sqrt{\frac{\log(2d)+t}{m}}
    \le\eps_0.
\]
Moreover, the definition of \(p\) gives
\[
    \log(2d)+t
    =\log\frac{2d}{\delta}
    \le\log\frac{8d}{\delta}
    \le\frac{p}{8},
\]
and hence
\[
    \frac{\log(2d)+t}{s}
    \le
    \frac{p/8}{p/(8\eps_0)}
    =\eps_0.
\]

Proposition~\ref{prop:mean-error} gives
\[
    \mu:=\E\lVert E\rVert_{\op}\le C_1\eps_0.
\]
Applying the tail bound in Proposition~\ref{prop:mean-to-moments} gives
\[
    \Pp\!\left\{
        \lVert E\rVert_{\op}>
        C_{\mathrm{conc}}(C_1+2)\eps_0
    \right\}
    \le e^{-t}
    =\delta.
\]
By the choice of \(C_0\),
\[
    C_{\mathrm{conc}}(C_1+2)\eps_0
    =\frac{C_{\mathrm{conc}}(C_1+2)}{C_0}\eps
    \le\eps.
\]
so
\[
    \Pp\{\lVert E\rVert_{\op}>\eps\}\le\delta.
\]

Finally,
\[
    m
    \le
    \frac{4C_0^2\bigl(d+\log(1/\delta)\bigr)}{\eps^2}
    =
    O\!\left(
        \frac{d+\log(1/\delta)}{\eps^2}
      \right),
\]
and
\[
    s
    \le\frac{p}{4\eps_0}
    =\frac{C_0p}{4\eps}
    =
    O\!\left(
        \frac{\log(d/\delta)}{\eps}
      \right).
\]
This proves the theorem.
\end{proof}

\clearpage
\bibliographystyle{plainnat}
\bibliography{refs}

\clearpage
\appendix

\section{Full proof of the restricted Eulerian contraction}\label{app:contraction}

We prove Lemma~\ref{lem:restricted-contraction}. For $A\subseteq[n]$, write
\begin{equation*}
    P_A=\sum_{i\in A}u_i u_i^\top.
\end{equation*}

\subsection{Three Parseval facts}

\begin{claim}[Restricted frame operator]\label{clm:restricted-frame}
For every $A\subseteq[n]$,
\[
    0\preceq P_A\preceq\Id.
\]
In particular, $\lVert u_i\rVert_2\le1$ for every $i$.
\end{claim}

\begin{proof}
Every $u_i u_i^\top$ is positive semidefinite, and
\[
    \Id-P_A=\sum_{i\notin A}u_i u_i^\top\succeq0.
\]
Thus $P_A\preceq\Id$. Taking $A=\{i\}$ gives $u_i u_i^\top\preceq\Id$, whose only nonzero eigenvalue is $\lVert u_i\rVert_2^2$.
\end{proof}

\begin{claim}[Restricted leaf sum]\label{clm:leaf-sum}
Let $A\subseteq[n]$, let $M\in\R^{d\times d}$ satisfy $\lVert M\rVert_{\op}\le1$, and let $x\in\R^d$. Then
\begin{equation*}
    \sum_{i\in A}|\langle u_i,Mx\rangle|^2
    =x^\top M^\top P_A Mx
    \le\lVert x\rVert_2^2.
\end{equation*}
The same bound holds with $\langle x,Mu_i\rangle$ in place of $\langle u_i,Mx\rangle$.
\end{claim}

\begin{proof}
By Claim~\ref{clm:restricted-frame},
\[
    x^\top M^\top P_A Mx
    \le x^\top M^\top Mx
    \le\lVert x\rVert_2^2.
\]
The second form follows by transposing $M$.
\end{proof}

\begin{claim}[Squared tree sum]\label{clm:squared-tree}
Let $\cT$ be a tree with at least one edge. Give each vertex $v$ an allowed set $A_v\subseteq[n]$, and each oriented edge $e=(v,w)$ a matrix $M_e$ with $\lVert M_e\rVert_{\op}\le1$. Then
\begin{equation*}
    \sum_{a(v)\in A_v}
    \prod_{e=(v,w)\in E(\cT)}
    |\langle u_{a(v)},M_eu_{a(w)}\rangle|^2
    \le d.
\end{equation*}
\end{claim}

\begin{proof}
Root the tree at a vertex $\rho$. Sum labels from the leaves toward the root. If $v$ is a current leaf with parent $w$, the only remaining edge factor involving $a(v)$ is
\[
    |\langle u_{a(v)},Mu_{a(w)}\rangle|^2,
\]
up to extra powers of $\lVert u_{a(v)}\rVert_2$ left by already removed children. These powers are at most one by Claim~\ref{clm:restricted-frame}, so dropping them can only increase the sum. Claim~\ref{clm:leaf-sum} replaces the sum over $a(v)$ by at most $\lVert u_{a(w)}\rVert_2^2$.

Continue until only the root remains. Because the original tree had an edge, at least one factor $\lVert u_{a(\rho)}\rVert_2^2$ remains. Any extra powers are at most one. The final sum is at most
\[
    \sum_{i=1}^n\lVert u_i\rVert_2^2
    =\tr\!\left(\sum_{i=1}^n u_i u_i^\top\right)
    =d.
\]
\end{proof}

\subsection{Connected contraction}

The label sum and edge product factor over connected components, so it is enough to prove a bound $d$ when $H$ is connected. We induct on the number of vertices. We use the Nash--Williams criterion: a finite multigraph contains two edge-disjoint spanning trees exactly when every partition of its vertex set into $t\ge2$ nonempty parts has at least $2(t-1)$ edges between different parts~\cite{NashWilliams1961}.

If $H$ has one vertex, it has at least one loop because there are no isolated vertices. Every loop factor satisfies
\[
    |\langle u_i,M_eu_i\rangle|
    \le\lVert u_i\rVert_2^2
    \le1.
\]
Since at least one loop is present, the product of all loop factors is at most $\lVert u_i\rVert_2^2$. Summing over $i$ gives at most $d$.

Now assume $H$ has at least two vertices.

\paragraph{Case 1: two edge-disjoint spanning trees.}
Let $\cT_1,\cT_2$ be such trees. For a fixed labeling, let $A$ and $B$ be the products of the edge factors on $\cT_1$ and $\cT_2$. Every factor outside $\cT_1\cup \cT_2$ has absolute value at most one, so the full product is at most $|AB|$. Since
\[
    2|AB|\le |A|^2+|B|^2,
\]
Claim~\ref{clm:squared-tree} bounds the sum of each squared-tree term by $d$. Hence the connected contribution is at most $d$.

\paragraph{Case 2: no two edge-disjoint spanning trees.}
We reduce across a two-edge cut.

\begin{claim}[Eulerian cuts are even]\label{clm:cuts-even}
For every $\Omega\subseteq V(H)$, let $\delta_{H}(\Omega)$ denote the set of edges with exactly one endpoint in $\Omega$. Then $|\delta_{H}(\Omega)|$ is even.
\end{claim}

\begin{proof}
The sum of the degrees in $\Omega$ is
\[
    2|E(H[\Omega])|+|\delta_{H}(\Omega)|.
\]
The left side and the first term are even.
\end{proof}

\begin{claim}[A two-edge cut exists]\label{clm:two-edge-cut}
A connected Eulerian multigraph with no two edge-disjoint spanning trees has a cut of size two.
\end{claim}

\begin{proof}
Suppose no such cut exists. By connectedness and Claim~\ref{clm:cuts-even}, every nontrivial shore has at least four boundary edges. For any partition into $t\ge2$ nonempty parts, the sum of the $t$ boundary sizes is at least $4t$. Each edge between parts is counted twice, so there are at least $2t\ge2(t-1)$ crossing edges. Nash--Williams would then give two edge-disjoint spanning trees, a contradiction.
\end{proof}

Choose an inclusion-minimal shore $\Omega$ of a two-edge cut. Let $\cJ=H[\Omega]$. The two cut incidences inside $\Omega$ occur at vertices $a,b\in\Omega$, possibly with $a=b$.

\begin{claim}[The minimal shore has two spanning trees]\label{clm:shore-trees}
The graph $\cJ$ contains two edge-disjoint spanning trees.
\end{claim}

\begin{proof}
Partition $\Omega$ into $t>1$ nonempty parts $P_1,\ldots,P_t$. Every $P_i$ is a proper nonempty subset of $\Omega$. Its boundary in $H$ is nonzero and even. If it had size two, it would be a smaller shore of a two-edge cut, contrary to the choice of $\Omega$. Thus $|\delta_{H}(P_i)|\ge4$.

Let $m_P$ be the number of edges of $\cJ$ crossing between the parts. These edges contribute twice to the sum of the boundary sizes, and the two edges leaving $\Omega$ contribute once each. Hence
\[
    2m_P+2
    =\sum_{i=1}^t|\delta_{H}(P_i)|
    \ge4t,
\]
so $m_P\ge2t-1\ge2(t-1)$. Nash--Williams applied to $\cJ$ gives two edge-disjoint spanning trees.
\end{proof}

\begin{claim}[The shore becomes one contraction edge]\label{clm:shore-edge}
After summing all labels inside $\Omega$, the two-terminal expression between the exterior endpoints of the cut is a bilinear form
\[
    F(x,y)=\langle x,M_\Omega y\rangle
\]
with $\lVert M_\Omega\rVert_{\op}\le1$.
\end{claim}

\begin{proof}
Orient the two cut edges so that their factors are
\[
    \langle x,M_0u_{\eta(a)}\rangle
    \quad\text{and}\quad
    \langle u_{\eta(b)},M_1y\rangle,
\]
where $\eta(v)$ is the label at $v\in\Omega$. By homogeneity, assume $\lVert x\rVert_2=\lVert y\rVert_2=1$. Let $\cT_1,\cT_2$ be the two edge-disjoint spanning trees from Claim~\ref{clm:shore-trees}. For a fixed interior labeling, define
\[
    A=\langle x,M_0u_{\eta(a)}\rangle\prod_{e\in \cT_1}z_e,
    \qquad
    B=\langle u_{\eta(b)},M_1y\rangle\prod_{e\in \cT_2}z_e,
\]
where $z_e$ is the generalized factor on $e$. The remaining internal factors have absolute value at most one. Therefore
\[
    |F(x,y)|
    \le\sum_{\text{labels in }\Omega}|AB|
    \le\frac12\sum_{\text{labels in }\Omega}(|A|^2+|B|^2).
\]
Root $\cT_1$ at $a$ and remove nonroot labels with Claim~\ref{clm:leaf-sum}. At the root, the boundary factor gives
\[
    \sum_{i\in A_a}|\langle x,M_0u_i\rangle|^2\le1.
\]
Thus $\sum|A|^2\le1$. Rooting $\cT_2$ at $b$ gives $\sum|B|^2\le1$. Hence $|F(x,y)|\le1$ for unit $x,y$.

The expression is bilinear because the exterior vectors occur only in the two boundary factors. Thus $F(x,y)=\langle x,M_\Omega y\rangle$ for some matrix $M_\Omega$, and the unit-vector bound is exactly $\lVert M_\Omega\rVert_{\op}\le1$.
\end{proof}

Replace the whole shore $\Omega$ by the single generalized edge carrying $M_\Omega$. The graph stays connected. At each exterior endpoint, the number of removed incidences equals the number of inserted incidences, so all degrees remain even. The number of vertices decreases strictly. By Claim~\ref{clm:shore-edge}, the label sum is unchanged and the new edge matrix still has operator norm at most one. Repeating the reduction eventually reaches either the one-vertex base case or Case~1. Both are bounded by $d$. Multiplying over the original connected components proves Lemma~\ref{lem:restricted-contraction}.

\section{The second-moment bound}\label{app:small-d}

We record a second-moment bound which is valid for every dimension. Take $\ell=2$. There are four trace positions $L_1,R_1,L_2,R_2$. Every admissible block has even size, and no block may contain both $L_t$ and $R_t$ because $i_t\ne j_t$. The only admissible partitions are
\[
    \bigl\{\{L_1,L_2\},\{R_1,R_2\}\bigr\}
    \quad\text{and}\quad
    \bigl\{\{L_1,R_2\},\{R_1,L_2\}\bigr\}.
\]
Their component-incidence multigraphs are, respectively, two parallel edges with $g_\pi=1$ and a two-edge tree with $g_\pi=0$, as shown in Figure~\ref{fig:small-d-supports}.

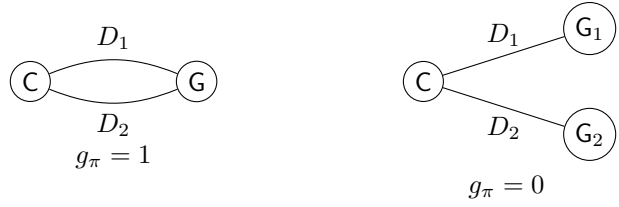
\begin{figure}[ht]
\centering
\begin{tikzpicture}[every node/.style={font=\small}]
\node[draw,circle,inner sep=2pt] (a1) at (0,0) {$\mathsf C$};
\node[draw,circle,inner sep=2pt] (b1) at (2.2,0) {$\mathsf G$};
\draw (a1) to[bend left=22] node[above] {$D_1$} (b1);
\draw (a1) to[bend right=22] node[below] {$D_2$} (b1);
\node at (1.1,-1.0) {$g_\pi=1$};

\node[draw,circle,inner sep=2pt] (a2) at (5.2,0) {$\mathsf C$};
\node[draw,circle,inner sep=2pt] (b21) at (7.4,0.7) {$\mathsf G_1$};
\node[draw,circle,inner sep=2pt] (b22) at (7.4,-0.7) {$\mathsf G_2$};
\draw (a2) -- node[above] {$D_1$} (b21);
\draw (a2) -- node[below] {$D_2$} (b22);
\node at (6.3,-1.4) {$g_\pi=0$};
\end{tikzpicture}
\caption{The only two component-incidence multigraphs that arise when $\ell=2$: a pair of parallel edges and a two-edge tree.}
\label{fig:small-d-supports}
\end{figure}

Applying \eqref{eq:fixed-partition-final} to these two shapes gives
\[
    \E\tr G^2\le2dp^2\le d(C_3p)^2.
\]
This proves the second-moment assertion of Proposition~\ref{prop:trace-moment}.

\section{The algebraic bucket influence lemma}\label{app:bucket-influence}

This appendix proves the algebraic inequality used in
\eqref{eq:bucket-influence-application} to control the squared positive
marginal contributions in the original bucket.

\begin{lemma}[Truncated positive influence]\label{lem:bucket-influence}
Let $t_1,\ldots,t_N\in\R$ and put
\[
    A=\sum_i t_i,
    \qquad
    w=\sum_i t_i^2,
    \qquad
    Q=A^2-w,
    \qquad
    K=w^2-\sum_i t_i^4.
\]
If $w\le\rho$, then
\begin{equation*}
    \sum_i\bigl(t_i(A-t_i)\bigr)_+^2
    =O\!\left(
        K+\min\bigl\{(Q)_+^2,\rho(Q)_+\bigr\}
    \right).
\end{equation*}
\end{lemma}

\begin{proof}
The claim is trivial when $w=0$. The expression is invariant under
replacing every $t_i$ by $-t_i$, so assume $A\ge0$.

First,
\begin{equation}\label{eq:bucket-linear-bound}
    \sum_i\bigl(t_i(A-t_i)\bigr)_+^2
    \le K+w(Q)_+.
\end{equation}
If $A^2\le w$, an active index satisfies
$0<t_i<A\le\sqrt w$, and
\[
    (A-t_i)^2
    \le(\sqrt w-t_i)^2
    \le w-t_i^2.
\]
Summing gives at most $K$. If $A^2\ge w$, then
$A\ge|t_i|$ and only positive $t_i$ contribute. Moreover,
\[
\begin{aligned}
    wA^2-\sum_i t_i^4
      -\sum_i\bigl(t_i(A-t_i)\bigr)_+^2
    &=
      \sum_{t_i<0}t_i^2(A^2-t_i^2)
      +2\sum_{t_i>0}t_i^3(A-t_i)\\
    &\ge0.
\end{aligned}
\]
Since $wA^2-\sum_i t_i^4=wQ+K$,
\eqref{eq:bucket-linear-bound} follows.

Second,
\begin{equation}\label{eq:bucket-quadratic-bound}
    \sum_i\bigl(t_i(A-t_i)\bigr)_+^2
    \le c\bigl(K+(Q)_+^2\bigr).
\end{equation}
For $Q\le0$, this follows from
\eqref{eq:bucket-linear-bound}. Suppose $Q>0$ and normalize
$w=1$. Choose $k$ with $|t_k|=\max_i|t_i|$, and put
\[
    \rho_0=\sum_{i\ne k}t_i^2.
\]
Since
\[
    \sum_i t_i^4
    \le\left(\max_i t_i^2\right)\sum_i t_i^2
    =1-\rho_0,
\]
we have $\rho_0\le K$. For $i\ne k$,
\[
    \sum_{i\ne k}t_i^2(A-t_i)^2
    \le2(1+Q)\rho_0+2\rho_0^2
    \le c(K+Q^2).
\]
If the $k$th influence is positive, then
\[
    2t_k(A-t_k)
    =Q+\rho_0-\left(\sum_{i\ne k}t_i\right)^2
    \le Q+\rho_0,
\]
so its square is also at most $c(Q^2+K)$. This proves
\eqref{eq:bucket-quadratic-bound}; rescaling restores general $w$.

Combining \eqref{eq:bucket-linear-bound} and
\eqref{eq:bucket-quadratic-bound} gives
\[
    \sum_i\bigl(t_i(A-t_i)\bigr)_+^2
    \le
    c\left[
        K+\min\bigl\{(Q)_+^2,w(Q)_+\bigr\}
    \right].
\]
Since $w\le\rho$, the result follows.
\end{proof}

For a SparseStack bucket, take $t_i=t_{\gamma i}$ for
$i\in I_\beta$. Then
\[
    A=A_\beta,
    \qquad
    w=\sum_{i\in I_\beta}t_{\gamma i}^2,
    \qquad
    Q=Q_\beta,
    \qquad
    K=K_\beta.
\]
By \eqref{eq:bucket-basic}, $w\le1/s$. Therefore
Lemma~\ref{lem:bucket-influence}, applied with $\rho=1/s$, gives
\eqref{eq:bucket-influence-application}.

\section{Standard probability tools}\label{app:probability-tools}

\subsection{Bousquet concentration for upper-bounded suprema}

\begin{lemma}[Bousquet, one-sided form]\label{lem:bousquet}
Let $X_1,\ldots,X_N$ be independent and identically distributed random variables
taking values in a common measurable space $\mathsf X$, and let
$\mathcal F$ be a countable class of measurable functions
$f:\mathsf X\to\R$. Suppose
\[
    \E f(X_i)=0,
    \qquad
    f(X_i)\le b
\]
almost surely for every $f\in\mathcal F$ and every $i$. Let
\[
    \mathcal Z=\sup_{f\in\mathcal F}\sum_{i=1}^N f(X_i),
    \qquad
    v=\sup_{f\in\mathcal F}\sum_{i=1}^N\E f(X_i)^2.
\]
Assume also that $\mathcal Z\ge0$, for example because the zero function
belongs to $\mathcal F$. Then, for every $t\ge0$,
\begin{equation*}
    \Pp\left\{
        \mathcal Z>\E\mathcal Z+\sqrt{2(v+2b\E\mathcal Z)t}+\frac{bt}{3}
    \right\}
    \le e^{-t}.
\end{equation*}
\end{lemma}

This is the one-sided, upper-bounded form of
Theorem~2.3 of Bousquet~\cite{Bousquet2002}, after scaling its
unit-envelope statement by $b$. Bousquet states the result for
i.i.d.\ observations, which is the form used here.

\subsection{Polynomial Efron--Stein}

\begin{lemma}[Polynomial Efron--Stein]\label{lem:polynomial-efron-stein}
Let $X_1,\ldots,X_N$ be independent, let
$F=f(X_1,\ldots,X_N)$, and let $F_i'$ be obtained by replacing
$X_i$ by an independent copy. Define
\[
    \mathcal V^+
    =\sum_{i=1}^N
      \E_i'\bigl[(F-F_i')_+^2\bigr].
\]
Then, for every real $q\ge2$,
\begin{equation*}
    \lVert(F-\E F)_+\rVert_q
    =O\!\left(\sqrt{q\,\lVert \mathcal V^+\rVert_{q/2}}\right).
\end{equation*}
\end{lemma}

This is Theorem~2 of
Boucheron--Bousquet--Lugosi--Massart~\cite{BoucheronEtAl2005}. Here \emph{polynomial} refers to power moments: the inequality controls the
\(L_q\)-norm of \((F-\mathbb{E}F)_+\), rather than an exponential moment. It is an
Efron--Stein inequality because the variance proxy \(\mathcal V^+\) is formed from the changes
\(F-F_i'\) produced by independently resampling one variable at a time.

\subsection{A matrix Bernstein moment corollary}

\begin{lemma}[Matrix Bernstein moment bound]\label{lem:matrix-bernstein-moment}
Let $M_1,\ldots,M_s$ be independent positive semidefinite matrices
of dimension $N$ such that
\[
    0\preceq M_\gamma\preceq I_N,
    \qquad
    \E M_\gamma\preceq B^{-1}I_N.
\]
Then, for every $q\ge2\log(2N)$,
\begin{equation*}
    \left\lVert
        \left\lVert\sum_{\gamma=1}^sM_\gamma\right\rVert_{\op}
    \right\rVert_q
    =O\!\left(\frac{s}{B}+q\right).
\end{equation*}
\end{lemma}

\begin{proof}
Set $H_\gamma=M_\gamma-\E M_\gamma$. Since
$M_\gamma^2\preceq M_\gamma$,
\[
    \sum_\gamma\E H_\gamma^2
    \preceq\sum_\gamma\E M_\gamma^2
    \preceq\frac{s}{B}I_N,
    \qquad
    \lambda_{\max}(H_\gamma)\le1.
\]
The self-adjoint matrix Bernstein inequality~\cite{Tropp2012}
gives
\[
    \Pp\left\{
        \lambda_{\max}\left(\sum_\gamma H_\gamma\right)\ge t
    \right\}
    \le
    N\exp\left(
        -\frac{t^2}{2(s/B+t/3)}
    \right).
\]
Integrating this scalar tail yields, for
$q\ge2\log(2N)$,
\[
    \left\lVert
        \left(
          \lambda_{\max}\left(\sum_\gamma H_\gamma\right)
        \right)_+
    \right\rVert_q
    \le c\left(\sqrt{\frac{qs}{B}}+q\right).
\]
Add
$\lVert\sum_\gamma\E M_\gamma\rVert_{\op}\le s/B$ and use
$\sqrt{(s/B)q}\le\tfrac12(s/B+q)$.
\end{proof}

\subsection{Symmetrization and contraction}

Let $\zeta_1,\ldots,\zeta_N$ be independent Rademacher signs. We use
the Rademacher contraction principle in moment form~\cite{LedouxTalagrand1991}:
if $\varphi(0)=0$ and $\varphi$ is $L$-Lipschitz, then for every
deterministic set $T\subseteq\R^N$ and every $q\ge1$,
\begin{equation}\label{eq:rademacher-contraction}
    \left\lVert
        \sup_{t\in T}
        \left|\sum_{j=1}^N\zeta_j\varphi(t_j)\right|
    \right\rVert_q
    \le
    2L
    \left\lVert
        \sup_{t\in T}
        \left|\sum_{j=1}^N\zeta_jt_j\right|
    \right\rVert_q.
\end{equation}

\begin{lemma}[Centered-function contraction]\label{lem:centered-contraction}
Let $X_1,\ldots,X_N$ be independent centered random functions on a
compact index set $T$. Let $\varphi(0)=0$ be $L$-Lipschitz. Then, for
every $q\ge1$,
\begin{equation*}
    \left\lVert
        \sup_{t\in T}
        \left|
          \sum_{j=1}^N
          \bigl(\varphi(X_j(t))-\E\varphi(X_j(t))\bigr)
        \right|
    \right\rVert_q
    \le
    8L
    \left\lVert
        \sup_{t\in T}
        \left|\sum_{j=1}^NX_j(t)\right|
    \right\rVert_q.
\end{equation*}
\end{lemma}

\begin{proof}
Let $X_1',\ldots,X_N'$ be independent copies. Jensen's inequality
and the usual moment symmetrization give
\[
    \left\lVert
        \sup_{t\in T}
        \left|
          \sum_j\bigl(\varphi(X_j(t))-\E\varphi(X_j(t))\bigr)
        \right|
    \right\rVert_q
    \le
    2
    \left\lVert
        \sup_{t\in T}
        \left|\sum_j\zeta_j\varphi(X_j(t))\right|
    \right\rVert_q.
\]
Conditioning on the functions and applying
\eqref{eq:rademacher-contraction} gives
\[
    \left\lVert
        \sup_{t\in T}
        \left|
          \sum_j\bigl(\varphi(X_j(t))-\E\varphi(X_j(t))\bigr)
        \right|
    \right\rVert_q
    \le
    4L
    \left\lVert
        \sup_{t\in T}
        \left|\sum_j\zeta_jX_j(t)\right|
    \right\rVert_q.
\]
Centering gives
$X_j=\E_{X'}(X_j-X_j')$. Jensen's inequality, followed by symmetry
of $X_j-X_j'$, allows the auxiliary signs to be absorbed:
\[
    \left\lVert
        \sup_{t\in T}
        \left|\sum_j\zeta_jX_j(t)\right|
    \right\rVert_q
    \le
    \left\lVert
        \sup_{t\in T}
        \left|\sum_j(X_j(t)-X_j'(t))\right|
    \right\rVert_q.
\]
The triangle inequality bounds the last expression by twice
\[
    \left\lVert
        \sup_{t\in T}
        \left|\sum_jX_j(t)\right|
    \right\rVert_q.
\]
\end{proof}

\subsection{Tail-to-moment conversion}

\begin{lemma}[Tail-to-moment conversion]\label{lem:tail-to-moment}
Let $X\ge0$ and suppose that, for every $t\ge0$,
\[
    \Pp\{X>a+b\sqrt t+ct\}\le e^{-t}.
\]
Then, for every $q\ge1$,
\begin{equation*}
    \lVert X\rVert_q
    \le a+O(b\sqrt q+cq).
\end{equation*}
The same bound holds for a finite sum of square-root terms by adding
their coefficients.
\end{lemma}

\begin{proof}
Let $T$ be an exponential random variable with mean one. The tail
assumption implies stochastic domination by $a+b\sqrt T+cT$. Since
\[
    \lVert T\rVert_q=O(q),
    \qquad
    \lVert\sqrt T\rVert_q=O(\sqrt q),
\]
Minkowski's inequality proves the claim.
\end{proof}

\end{document}